\documentclass[11pt]{article}

\usepackage{fullpage}
\usepackage[utf8]{inputenc}
\usepackage[T1]{fontenc}
\usepackage{textcomp}
\usepackage{parskip}
\usepackage{natbib}
\usepackage{hyperref}
\usepackage{url}
\usepackage{booktabs}
\usepackage{amsfonts}
\usepackage{amsmath}
\usepackage{verbatim}
\usepackage{nicefrac}
\usepackage{microtype} 
\usepackage[table]{xcolor} 
\usepackage{amsmath}
\usepackage{amsthm}
\usepackage{bm}
\usepackage{cleveref}
\usepackage{tikz}
\usepackage{enumitem}
\usepackage[ruled]{algorithm2e}
\usepackage{setspace}
\usepackage{subcaption}
\usepackage{multirow}
\usepackage{autonum}
\usepackage{enumitem}
\usepackage{float}

\def\Z{\mathcal{Z}}
\def\F{\mathcal{F}}

\def\G{\mathcal{G}}

\def\R{\mathbb{R}}
\def\P{\mathcal{P}}
\def\P2{\mathcal{P}_2}
\def\W2{\mathcal{W}_2}

\def\E{\mathbb{E}}
\def\V{\mathbb{V}}

\newtheorem{theorem}{Theorem}
\newtheorem{proposition}{Proposition}

\theoremstyle{definition}

\newtheorem{remark}{Remark}

\newcounter{mntcomm}

\newcommand{\lr}[1]{\left(#1 \right)}

\newcommand{\lrsb}[1]{\left[#1 \right]}

\title{Bures-Wasserstein Importance-Weighted Evidence Lower Bound: Exposition and Applications\footnote{Acknowledgement: We are grateful to the Editor, Associate Editor and the three anonymous reviewers for their careful reading of the manuscript and for their constructive comments and suggestions, which have substantially improved the paper.
}}
\author{Peiwen Jiang, Takuo Matsubara, and Minh-Ngoc Tran}
\date{The University of Sydney Business School, Australia}

\begin{document}

\maketitle

\begin{abstract}
The importance-weighted evidence lower bound (IW-ELBO) has emerged as a compelling objective for variational inference (VI), providing a tighter bound than the standard ELBO and mitigating its pathological mode-seeking behavior.
However, optimizing the IW-ELBO in Euclidean space can be inefficient:~as the number of importance samples $K$ increases, standard Euclidean gradient estimators suffer from a vanishing signal-to-noise ratio (SNR).
This paper reformulates the optimization of the IW-ELBO within the Bures-Wasserstein (BW) space, the manifold of Gaussian distributions equipped with the 2-Wasserstein metric.
We derive the infinite-dimensional Wasserstein gradient of the IW-ELBO and project it onto the BW space, yielding a computationally tractable BW gradient for Gaussian VI.
While the SNR of the Euclidean gradient is known to vanish at a rate of $\mathcal{O}(1/\sqrt{K})$, we prove that the SNR of the Wasserstein gradient scales favorably at a rate of $\Omega(\sqrt{K})$.
Crucially, the BW gradient maintains computational tractability while inheriting the stability of the Wasserstein geometry.
We establish that its SNR interpolates between the Euclidean and Wasserstein regimes, achieving a non-vanishing rate of $\Omega(1)$.
We extend this geometric analysis to the variational R\'{e}nyi importance-weighted autoencoder bound, proving analogous stability guarantees.
Empirical evaluations demonstrate that the proposed algorithm achieves superior mass-covering and algorithmic performance compared to established baselines.

\vspace{5pt}
\end{abstract}

\section{Introduction}

Driven by the growing sophistication of statistical models, posterior distributions in modern applications have increasingly become high-dimensional and complex.
This computational burden has motivated a shift from classical sampling techniques toward scalable alternatives. 
Two prominent paradigms have emerged at the forefront of this transition: variational inference (VI) and Wasserstein gradient flow (WGF).
VI \citep{Jordan:ML1999} frames posterior approximation as an optimization problem over the parameter of a tractable density family, a task typically achieved by maximizing the evidence lower bound (ELBO).
Gaussian VI, for example, is widely employed to optimize the mean and covariance within a Gaussian family.
In contrast, WGF \citep{Jordan1998} offers a non-parametric, geometric framework, under which the target-density approximation is formally characterized as an infinite-dimensional optimization over the space of probability distributions.

Recent literature has made substantial progress in unifying the two distinct paradigms of parametric VI and non-parametric WGF. 
Notably, \cite{Lambert2022} established that Gaussian VI can be recast as a WGF restricted to the manifold of Gaussian distributions. 
This unification relies on equipping the space of Gaussian measures with the 2-Wasserstein metric, giving rise to the Bures-Wasserstein (BW) space \citep{Bures1969}.
The BW geometry elegantly bridges the differential calculus in the parameter space with the infinite-dimensional optimization landscape over probability distributions. 
This geometric perspective facilitates novel theoretical insights and algorithmic developments by importing convex optimization techniques from optimal transport into VI. 
Leveraging this framework, \cite{Diao2023} developed an efficient algorithm for Gaussian VI that adapts a forward-backward scheme originally designed for WGF.
They demonstrated that exploiting the BW geometry yields the fastest known convergence rates for Gaussian VI.

Parallel to these developments, the importance-weighted evidence lower bound (IW-ELBO) has emerged as a critical advancement within parametric VI \citep{Burda2016,Domke2018}.
A well-documented limitation of the standard ELBO is its inherent mode-seeking behavior, which frequently causes the variational approximation to underestimate the tails of the target posterior. 
The IW-ELBO mitigates this deficiency by constructing a strictly tighter lower bound on the marginal log-likelihood, utilizing multiple samples from the variational distribution.
As the number of samples increases, the IW-ELBO monotonically approaches the true marginal likelihood. 
The statistical properties of the IW-ELBO have been extensively investigated in recent literature \citep{rainforth2018tighter,tan2020conditionally,daudel2023alpha,daudel2026importance}.

Despite the empirical and theoretical successes of the IW-ELBO within parametric VI, its geometric interpretation and potential connection to the WGF framework remain entirely unestablished. 
Establishing this link is of significant methodological interest, as synthesizing the tightened evidence bounds of the IW-ELBO with the geometric machinery of WGF provides a pathway to robust, scalable inference algorithms that circumvent the pathologies of standard Gaussian VI. This paper formally bridges this gap by advancing the intersection of importance-weighted variational bounds and the BW geometry. Our main contributions are summarized as follows:

\paragraph{Wasserstein and BW Gradient} We derive the infinite-dimensional Wasserstein gradient of the IW-ELBO, characterizing the steepest ascent direction over the space of probability measures and elucidating its structural distinction from Euclidean gradients. To address the inherent intractability of the Wasserstein gradient, we subsequently derive its projection onto the BW manifold.
The resulting BW gradient yields a closed-form, computationally tractable gradient of the IW-ELBO for Gaussian VI under the Wasserstein geometry.

\paragraph{Non-Degenerate Signal-to-Noise Ratio} A well-documented pathology in optimizing the IW-ELBO is that the signal-to-noise ratio (SNR) of standard Euclidean gradient estimators degenerates as the number of copies, $K$, increases \citep{rainforth2018tighter}.
The SNR decays as $\mathcal{O}(1 / \sqrt{K})$, meaning that the standard deviation of the gradient estimator diverges to infinity, relative to the magnitude of the exact gradient, as $K$ increases.

We formally prove that the Wasserstein and BW gradient of the IW-ELBO circumvents this degeneration. 
We establish that the SNR of the Wasserstein gradient scales favorably at a rate of $\Omega(\sqrt{K})$, guaranteeing stabilized gradient estimation in the large-$K$ regime.
The BW gradient enjoys the computational tractability while inheriting this stability.
Its SNR interpolates those of the Wasserstein and Euclidean gradients, yielding a rate of $\Omega(1)$ constant across all values of $K$.

\paragraph{BW IW-ELBO Optimization} We develop an efficient and accurate Gaussian VI based on the IW-ELBO, leveraging the derived BW gradients and the Euler discretization of WGF. We demonstrate strong mass-covering capabilities of our variational approximation, effectively capturing target density tails and mitigating the pathological underestimation common for standard Gaussian VI. We assess the efficiency of our algorithm through comprehensive empirical evaluations. 

\paragraph{Extension to Generalized Bounds} Finally, we extend our methodology and its analysis to the variational R\'{e}nyi importance-weighted autoencoder (VR-IWAE) bound \citep{daudel2023alpha}, another recent generalization of the ELBO. We derive the corresponding Wasserstein and BW gradients for the VR-IWAE bound, which exhibit the same SNR scaling rates as those for the IW-ELBO. 
We then provide a complementary BW Gaussian VI algorithm for the VR-IWAE bound.

Our paper sits within the rapidly growing body of literature exploiting optimal transport (OT) in statistics and machine learning. One of the earliest and most influential applications of OT is in distributional data analysis, where each observation is a probability distribution and Wasserstein geometry provides a natural metric for comparing, averaging, and performing statistical inference on such objects \citep{chen2023wasserstein,ghodrati2022distribution,matsubara2026wasserstein}. Beyond distributional data, OT has become a fundamental tool in statistical theory, where Wasserstein distances have been used to establish identifiability, convergence rates, and asymptotic theory for latent structured models, mixture models, and Bayesian nonparametric procedures \citep{nguyen2026optimal,catalano2021measuring}. In computational statistics and machine learning, OT has also found widespread use in applications including generative modeling and Bayesian computation, owing to its ability to capture the underlying geometry of probability distributions. More recently, there has been increasing interest in exploiting the differential geometry of Wasserstein space as a mathematical framework for optimization over probability distributions \citep{salim2020wasserstein,Diao2023,Lambert2022}. Our work contributes to this latter direction by further developing optimization algorithms that exploit the intrinsic geometry of BW space for statistical learning.

The remainder of the paper is organized as follows.
\Cref{sec:preliminaries} reviews the necessary background on VI and WGF.
\Cref{eq:BW-IW-ELBO} derives the Wasserstein and Bures-Wasserstein gradients of the IW-ELBO, formulating the objective as a functional over probability distributions.
We then investigate the theoretical properties of these gradients,
before deriving the practical Gaussian VI algorithm, leveraging the WGF of IW-ELBO with BW geometry.
\Cref{sec:Numerical applications} presents empirical experiments validating the performance of our framework.
Finally, \Cref{sec:VR-IWAE} extends our analysis to the VR-IWAE bound, followed by concluding remarks in \Cref{sec:conclusion}.
Proofs and technical details are included in the Appendix.

\section{Preliminaries} \label{sec:preliminaries}

This section provides some preliminaries of VI, the IW-ELBO, and Wasserstein geometry required for our methodological development.
We first introduce a set of standard notations.

\paragraph{Setup and Notation}
Let $\| \cdot \|$ and $\langle\cdot,\cdot\rangle$ denote the standard norm and inner-product in $\mathbb{R}^d$. 
We denote by $\mathbf{S}^d$ the space of $d \times d$ symmetric matrices.
Let $\mathcal{P}_2(\mathbb{R}^d)$ be the space of probability measures on $\mathbb{R}^d$ with a finite second moment. 
For a measurable map $T:\mathbb{R}^d\to\mathbb{R}^d$, $T_{\#}\mu$ denotes the push-forward measure of $\mu$ through $T$. 
For sequences of positive numbers $a_n$ and $b_n$, we write $a_n=O(b_n)$ ($a_n=\Omega(b_n)$) if there exists a constant $c>0$ such that $a_n/b_n\leq c$ ($a_n/b_n\geq c$, respectively) for all $n$ sufficiently large.
Similarly, we write $a_n=o(b_n)$ if $a_n/b_n\to 0$ as $n\to\infty$.
 
\subsection{Variational Inference and Importance-Weighted ELBO}

Many statistical models involve unobserved latent variables $z\in\mathcal{Z}$ to explain the structure within observed data $x$, postulating a joint distribution $p(x,z)$. 
VI offers a scalable framework to approximate the analytically intractable posterior $p(z|x)$ using a tractable variational distribution $q_\psi(z)$ parameterized by $\psi$. 
Standard VI optimizes the variational parameter $\psi$ by maximizing the ELBO, derived as a lower bound on the log marginal likelihood $\log p(x)$:
\begin{align}
    \operatorname{ELBO}(q_\psi) := \mathbb{E}_{z \sim q_\psi}\left[ \log \left( \frac{ p(x, z) }{ q_\psi(z) } \right) \right] \leq \log p(x) .
\end{align}
This lower bound induces the fundamental decomposition $\log p(x) = \operatorname{ELBO}(q_\psi) + \text{KL}(q_\psi(z) || p(z | x))$, where the latter term is the Kullback-Leibler (KL) divergence from the variational distribution to the posterior. Because the log marginal likelihood $\log p(x)$ is constant with respect to $\psi$, maximizing the ELBO is equivalent to minimizing this KL divergence. However, because the KL divergence is the expectation of $\log q_\psi(z)-\log p(z|x)$ taken with respect to $q_\psi(z)$, it places disproportionate importance on the log-density mismatch in high-probability regions of the variational distribution. Consequently, maximizing the ELBO might not enforce accurate estimation of the posterior tail probabilities; rather, it frequently prioritizes finding a single high-probability mode of a potentially multimodal posterior $p(z|x)$, leading to a well-documented systematic underestimation of the posterior variance.

To address these limitations of the standard ELBO, several advanced methods have been developed, among which the IW-ELBO of \cite{Burda2016} remains one of the most prominent.
Given a specified variational distribution, the IW-ELBO constructs a more accurate estimate of the marginal log-likelihood by leveraging importance weights. 
Utilizing $K$ independent and identically distributed random variables $z_1, \dots, z_K \sim q_\psi(z)$, the objective is defined as:
\begin{align}
    \operatorname{IW-ELBO}_K(q_\psi) := \mathbb{E}_{z_1, \dots, z_K \sim q_\psi}\left[ \log \left(\frac{1}{K}\sum_{k=1}^K \frac{p(x,z_k)}{q_\psi(z_k)} \right) \right] . \label{eq:parametric_IW_ELBO}
\end{align}
When $K=1$, this formulation reduces to the standard ELBO. By Jensen's inequality, the IW-ELBO is shown to interpolate monotonically between the ELBO and the marginal log-likelihood as a function of the sample size $K$ \citep{Burda2016}:
\begin{align}
    \log p(x) \ge \operatorname{IW-ELBO}_{K+1}(q_\psi) \ge \operatorname{IW-ELBO}_{K}(q_\psi) \ge \operatorname{ELBO}(q_\psi),
\end{align}
and converge to the log marginal likelihood $\log p(x)$ as $K \to \infty$. 
Consequently, given sufficient computational resources, the IW-ELBO provides an arbitrarily tight lower bound, a property that underpins its documented superior performance over standard VI across a range of applications \citep{Burda2016, mnih2016variational, Tucker2019DoublyRG, daudel2023alpha, daudel2026importance}.

\subsection{Wasserstein and Bures-Wasserstein Geometry} \label{sec:bures_wasserstein_geometry}

The Wasserstein space provides a natural geometric framework for modeling the dynamics of probability distributions. The 2-Wasserstein metric between two distributions $\mu, \nu \in \mathcal{P}_2(\mathbb{R}^d)$ is defined as:
\begin{align}
    W_2(\mu,\nu) := \left\{\inf_{T:T_{\#}\mu=\nu}\int_{\mathbb{R}^d}\|x-T(x)\|^2\mu(dx)\right\}^{1/2}.
\end{align}
This metric space, denoted $\mathbb{W}_2(\mathbb{R}^d)$, is endowed with a rich differential structure \citep{villani2009optimal,Ambrosio2005}. 
Under its geometry, the notion of gradient for a functional $\mathcal{F}$ defined on $\mathbb{W}_2(\mathbb{R}^d)$ can be derived using the first variation \citep{Santambrogio:OTbook}. 
The first variation of $\mathcal{F}$ at a distribution $\mu$ is a function $\delta \mathcal{F}(\mu) / \delta \mu: \mathbb{R}^d \to \mathbb{R}$ satisfying the following equation:
\begin{align}
    \lim_{\epsilon \to 0^+} \frac{\mathcal{F}(\mu + \epsilon \nu ) - \mathcal{F}(\mu) }{\epsilon} = \int_{\mathbb{R}^d} \frac{\delta \mathcal{F}(\mu)}{\delta \mu} (x) \nu(dx) \label{eq:first_variation}
\end{align}
for all signed measures $\nu$ such that $\mu + \epsilon \nu \in \mathcal{P}_2(\mathbb{R}^d)$ for sufficiently small $\epsilon > 0$. 
The Wasserstein gradient of $\mathcal{F}$ at $\mu$, denoted $\nabla^\mathrm{W} \mathcal{F}(\mu): \mathbb{R}^d \to \mathbb{R}^d$, is defined as the Euclidean gradient of this first variation:
\begin{align}
    [\nabla^\mathrm{W} \mathcal{F}(\mu)](x) = \nabla_x \frac{\delta \mathcal{F}(\mu)}{\delta\mu}(x) .
\end{align}
The Wasserstein gradient formally represents the direction in which the functional $\mathcal{F}$ increases most rapidly under the intrinsic geometry of $\mathbb{W}_2(\mathbb{R}^d)$ \citep{Ambrosio2005}.

Albeit the elegant theory, performing optimization within the Wasserstein space presents significant computational challenges. 
For example, the Wasserstein gradient can be intractable even for common functionals $\mathcal{F}$ such as the KL divergence.
A powerful strategy to ensure tractability is to restrict focus to the space of non-degenerate Gaussian distributions equipped with the 2-Wasserstein metric, formally termed the Bures-Wasserstein (BW) space \citep{Lambert2022}. 
Denoted by $\operatorname{BW}(\mathbb{R}^d)$, this space constitutes a Riemannian submanifold of $\mathbb{W}_2(\mathbb{R}^d)$ wherein calculus operations become highly intuitive and computationally efficient. 
In this manifold, the tangent space at a given Gaussian distribution $q = \mathcal{N}(m, \Sigma)$ is explicitly identified as a set of affine maps \citep{Diao2023}:
\begin{align}
    \mathcal{T}_q \operatorname{BW}(\mathbb{R}^d) := \Big\{ x \mapsto a + S (x - m) \mid a \in \mathbb{R}^d, S \in \mathbf{S}^d \Big\},
\end{align}
with $\mathbf{S}^d$ the set of symmetric matrices. 

Because $\operatorname{BW}(\mathbb{R}^d)$ is an embedded submanifold of $\mathbb{W}_2(\mathbb{R}^d)$, the Riemannian gradient of a functional $\mathcal{F}$ on $\operatorname{BW}(\mathbb{R}^d)$, referred to as the BW gradient, can be analytically computed by projecting the full Wasserstein gradient onto the BW tangent space.
Specifically, at a given Gaussian distribution $q = \mathcal{N}(m, \Sigma)$, the BW gradient $\nabla^{\operatorname{BW}} \mathcal{F}(q)$ is defined via the orthogonal projection:
\begin{align}
    \nabla^{\operatorname{BW}} \mathcal{F}(q) := \text{arg}\min_{v \in \mathcal{T}_q \operatorname{BW}(\mathbb{R}^d)} \| v - \nabla^\mathrm{W} \mathcal{F}(q) \|_{L_q^2(\mathbb{R}^d)} . \label{eq:BW gradient derivation}
\end{align}
Because elements of the tangent space $\mathcal{T}_q \operatorname{BW}(\mathbb{R}^d)$ are strictly affine maps, the resulting BW gradient has the affine form $[\nabla^{\operatorname{BW}} \mathcal{F}(q)](x) = a_* + S_*(x - m)$, where 
\begin{align}
    a_* := \mathbb{E}_{X \sim q}[ G(X) ] \qquad \text{and} \qquad S_* := \mathbb{E}_{X \sim q}[ \nabla G(X) ] . \label{eq:BW gradient derivation}
\end{align}
with $G(x) := [\nabla^\mathrm{W} \mathcal{F}(q)](x)$.
See also \cite{Diao2023} for more detail.

\section{Bures-Wasserstein IW-ELBO} \label{eq:BW-IW-ELBO}

To date, standard optimization of the IW-ELBO relies on Euclidean gradient descent over the parameters $\psi$. 
This, however, introduces a fundamental ``objective-mechanism mismatch.''
While the VI objective is essentially a functional over probability distributions, Euclidean gradient descent does not leverage their intrinsic geometry.
The resulting optimization dynamics over distributions can often be unstable and inefficient. 
Aligning the optimization mechanism with the meaningful geometry of probability distributions directly motivates the adoption of the Wasserstein space.

This section presents our main results concerning the IW-ELBO objective for Gaussian VI under the Wasserstein geometry.
Proofs of all the theoretical results are provided in \Cref{apx:proofs}.

\subsection{Wasserstein and BW Gradient of IW-ELBO}\label{sec:WG_IW_ELBO} 

We begin by providing a rigorous derivation of the Wasserstein and BW gradient for the IW-ELBO.
Recall that the Wasserstein gradient is defined over functionals of probability distributions.
To this end, we reformulate the IW-ELBO as a functional of $K$ probability distributions $q_1, \dots, q_K$: 
\begin{equation}
	\operatorname{IW-ELBO}_K (q_1, \dots, q_K) := \mathop{\E}_{z_1 \sim q_1, \dots, z_K \sim q_K}\left[ \log \left( \frac{1}{K} \sum_{i=1}^{K} \frac{p(x, z_i)}{q_i(z_i)} \right) \right] . \label{eq:IW_ELBO_Functional}
\end{equation}
The standard IW-ELBO is recovered by substituting a parametric variational distribution $q_\psi$ for all the $K$ arguments $q_1, \dots, q_K$.
In contrast, we consider a configuration where all $K$ arguments $q_1, \dots, q_K$ are set to a common arbitrary distribution $q$.

Our first aim is to derive the coordinatewise Wasserstein gradient of the IW-ELBO with respect to the $n$-th argument $q_n$ for an arbitrary index $n\in\{1,...,K\}$.
Subsequently, we demonstrate that the functional form of this gradient does not depend on the choice of index $n$.
Given an index $n$, fixing the remaining $K-1$ arguments $\{ q_i \}_{i \ne n}$ to the common distribution $q$, we define a functional $\F_n(q_n)$ on $\mathbb{W}_2(\R^d)$ by
\begin{equation}
    \F_n(q_n) := \mathop{\E}_{z_n \sim q_n} \left[\mathop{\E}_{ \{ z_1, \dots, z_K \} \setminus z_n \overset{i.i.d.}{\sim} q } \left[ \log \left( \frac{1}{K} \sum_{i = 1}^{K} \frac{p(x, z_i)}{q_i(z_i)} \right) \right] \right] ,
\end{equation}
where the subscript $\{ z_1, \dots, z_K \} \setminus z_n \overset{i.i.d.}{\sim} q$ denotes that the variables $\{ z_i \}_{i \neq n}$ are i.i.d.~random variables following $q$.
This expression represents the IW-ELBO viewed as a functional of the $n$-th argument $q_n$.
The coordinatewise Wasserstein gradient at $q_n = q$ is established as below.

\begin{proposition}[Coordinatewise Wasserstein Gradient] \label{pro: proposition1}
Assume standard regularity conditions that justify interchanging differentiation and expectation.
The coordinatewise Wasserstein gradient of $\F_n(q_n)$ at $q_n = q$ is given by
\begin{equation}\label{eq:Wasserstein grad IW-ELBO}
	\nabla^W \left[ \F_n(q) \right](z_n)  = \mathop{\E}_{ \{ z_1, \dots, z_K \} \setminus z_n \overset{i.i.d.}{\sim} q } \left[ \left( \frac{ w(z_n) }{ \sum_{i=1}^{K} w(z_i)} \right)^2 \nabla_{z_n} \log w(z_n) \right] \in \R^d ,
\end{equation}    
where we define $w(z) := p(x, z) / q(z)$.
\end{proposition}

The common intractability of the Wasserstein gradient stems from the weight function $w(z) = p(x, z) / q(z)$ depending on the density $q$, which is not necessarily available in nonparametric settings.
This intractability can be resolved by restricting focus to the Gaussian variational distributions.

Building upon the above derivation, we now establish the BW gradient of the IW-ELBO.
To this end, we first substitute a Gaussian distribution $\mathcal{N}(m, \Sigma)$ into the argument $q$ of the coordinatewise Wasserstein gradient established in \eqref{eq:Wasserstein grad IW-ELBO}.
We then take a projection of the resulting Wasserstein gradient \eqref{eq:Wasserstein grad IW-ELBO} onto the tangent space of the BW space, yielding the following proposition.

\begin{proposition}[BW Gradient] \label{pro: proposition2}
Under the same condition as \Cref{pro: proposition1}, the BW gradient of $\mathcal{F}_n(q)$ at a Gaussian density $q = \mathcal{N}(m, \Sigma)$ is given by the following affine map
\begin{align}
    \nabla^{\emph{BW}}[ \mathcal{F}_n(q) ](z) = a_* + S_* (z - m) . \label{eq:bw_iw_elbo0}
\end{align}	
The vector $a_*$ and matrix $S_*$ are defined as, respectively,
\begin{align}
    a_* & := \mathop{\E}_{z_1, \dots, z_K \overset{i.i.d.}{\sim} \mathcal{N}(m,\Sigma)} \left[ \left( \frac{ w(z_n) }{ \sum_{i=1}^{K} w(z_i)} \right)^2 \nabla_{z_n} \log w(z_n) \right] ; \label{eq:bw_iw_elbo1} \\
    S_* & := \mathop{\E}_{z_1, \dots, z_K \overset{i.i.d.}{\sim} \mathcal{N}(m,\Sigma)} \left[ \nabla_{z_n}\left\{ \left(\frac{ w(z_n) }{ \sum_{i=1}^{K} w(z_i)} \right)^2 \nabla_{z_n} \log w(z_n)\right\} \right] , \label{eq:bw_iw_elbo2}
\end{align}	
where we define $w(z) := p(x, z) / q(z)$ with the Gaussian density $q = \mathcal{N}(m, \Sigma)$.
\end{proposition}

The BW gradient is readily computable via the Monte Carlo estimation as in \eqref{eq:bw_iw_elbo1}--\eqref{eq:bw_iw_elbo2}, because the weight function $w(z) = p(x, z) / q(z)$ is entirely tractable for the Gaussian density $q = \mathcal{N}(m, \Sigma)$.

\begin{remark}[Coordinate Invariance] \label{rm: invariance}
    The expressions for the Wasserstein gradient \eqref{eq:Wasserstein grad IW-ELBO} and the BW gradient \eqref{eq:bw_iw_elbo0} are invariant to the coordinate index $n$.
    Consider the BW gradient \eqref{eq:bw_iw_elbo0} first.
    Since $z_1, \dots, z_K$ are i.i.d.~random variables and the weight function $w(z)$ does not depend on the index $n$, the expected values for $a_*$ and $S_*$ are invariant under permutations of $z_1, \dots, z_K$.
    A similar argument holds for the Wasserstein gradient, where the notation $z_n$ represents the location that the gradient is evaluated and its subscript $n$ serves merely as a notational placeholder.
\end{remark}

By permutation symmetry, the full Wasserstein gradient of the joint objective is $K$ times the coordinatewise gradient.
The SNR, which we analyze next, remains invariant to constant scaling of the gradient since it is defined via a ratio.
Furthermore, constant scaling does not affect the subsequent algorithmic procedure as it is absorbed into the learning rate.
It therefore suffices to evaluate the coordinatewise gradient for a single arbitrary index, setting $n = K$ in what follows.

\subsection{Non-Degenerate Signal-to-Noise Ratio of Wasserstein and BW Gradient}\label{eq:Euclidean vs Wasserstein}

The IW-ELBO provides an arbitrarily tight lower bound of the marginal likelihood as $K$ increases.
However, \cite{rainforth2018tighter} observed that using large $K$ can computationally hinder the optimization of the variational parameters $\psi$.
We first recap their argument.

Recall the definition of the IW-ELBO in \eqref{eq:parametric_IW_ELBO}.
Here, the variational family $q_\psi$ is indexed by an $s$-dimensional parameter vector $\psi=(\psi_1, \dots, \psi_s)$.
Assume that the variational distribution $q_\psi$ can be expressed as the push-forward distribution $q_\psi = ( T_\psi )_{\#} q_\epsilon$ from some base distribution $q_\epsilon$ and transform $T_\psi$ parameterized by $\psi$.
Under this reparameterization, the Euclidean gradient of the IW-ELBO with respect to the $r$-th coordinate $\psi_r$ is given by
\begin{align}\label{eq:Euclidean grad IW-ELBO}
	\frac{\partial}{\partial\psi_r}\operatorname{IW-ELBO}_K(q_\psi) = \mathop{\E}_{\epsilon_1, \dots, \epsilon_K \overset{i.i.d.}{\sim}  q_\epsilon(\cdot)} \left[\sum_{i=1}^K \frac{ w_\psi\big(T_\psi(\epsilon_i)\big) }{ \sum_{k=1}^{K} w_\psi\big(T_\psi(\epsilon_k)\big)}  \frac{\partial}{\partial\psi_r} \log w_\psi\big(T_\psi(\epsilon_i)\big) \right],
\end{align}
where we define $w_\psi(z) := p(x,z)/q_\psi(z)$ for brevity.
Let $g^r_{M,K}(\psi)$ denote the Monte Carlo estimator of the gradient \eqref{eq:Euclidean grad IW-ELBO} at $\psi$ constructed from $M$ independent samples of the $K$ random variables $\epsilon_1, \dots, \epsilon_K$.
\cite{rainforth2018tighter} showed that the SNR of this estimator scales as:
\begin{equation}
    \text{SNR}\big(g^{r}_{M,K}(\psi)\big)=\frac{|\E\big( g^{r}_{M,K}(\psi)\big)|}{\sqrt{\V(g^{r}_{M,K}(\psi))}} = O\left( \frac{\sqrt{M}}{\sqrt{K}} \right) .
\end{equation}
The SNR degenerates as $K$ increases, meaning that the relative standard deviation of the Monte Carlo estimator with respect to the magnitude of the exact gradient diverges to infinity.
This implies that accurate estimation of the gradient \eqref{eq:Euclidean grad IW-ELBO} becomes increasingly challenging for large $K$.

Observe that the Wasserstein gradient \eqref{eq:Wasserstein grad IW-ELBO} of the IW-ELBO features a squared normalized weight term $w(z_n)^2 / ( \sum_{i=1}^{K} w(z_i) )^2$.
This is in stark contrast to the linear term appearing in the Euclidean gradient \eqref{eq:Euclidean grad IW-ELBO}.
The squared expression bears a close resemblance to the \emph{doubly-reparameterized gradient} of the IW-ELBO, derived in \cite{Tucker2019DoublyRG}.
They demonstrated that the Monte Carlo estimator for the doubly-reparameterized gradient exhibits an SNR scaling of $O(\sqrt{MK})$, effectively mitigating the deterioration shown for the standard gradient.
Remarkably, the Wasserstein gradient \eqref{eq:Wasserstein grad IW-ELBO} of the IW-ELBO exhibits this same favorable SNR scaling.
This implies that the precision of the Wasserstein gradient estimator improves as $K$ increases.

We now formally establish the SNR scaling rate of the Wasserstein gradient \eqref{eq:Wasserstein grad IW-ELBO}.
As in \Cref{rm: invariance}, since this Wasserstein gradient is invariant to the coordinate index $n$, it suffices to focus on $n = K$.
Consider the Monte Carlo estimation of this gradient constructed from $M$ independent samples of the $K-1$ variables $z_1, \dots, z_{K-1}$.
Here, the coordinatewise Wasserstein gradient \eqref{eq:Wasserstein grad IW-ELBO} evaluated at an arbitrary location $z$ is a $d$-dimensional vector.
Similarly to \cite{rainforth2018tighter}, we consider the SNR of the $r$-th coordinate of the Wasserstein gradient at $z$ to streamline our analysis.
By abuse of notation, let $g^{r}_{M,K}(z)$ denote the $r$-th coordinate of this Monte Carlo estimator.

\begin{theorem}\label{the: Wass grad SNR} Assume that (i) $\E_{Z\sim q}(w(Z)^4)<\infty$ and that (ii) $\E_{Z\sim q}(w(Z)^{-12})<\infty$. Then, for any location $z$ at which $w(z)>0$, we have
\begin{equation}
\text{SNR}\big(g^{r}_{M,K}(z)\big) = \Omega(\sqrt{MK}) .
\end{equation}
\end{theorem}

We emphasize that this SNR enjoys a {\it Big-Omega} rate of $\sqrt{MK}$, strictly lower-bounded by a term that grows with $K$.
This result sheds a light on the appealing property of the Wasserstein gradient, although its computation is typically intractable.
We now turn our attention to the BW gradient.

Remarkably, the SNR of the BW gradient interpolates between these of the Euclidean and Wasserstein gradients.
Consider the Monte Carlo estimation of the BW gradient \eqref{eq:bw_iw_elbo0}, where the corresponding terms $a_*$ and $S_*$ are estimated from $M$ independent samples of the $K$ random variables $z_1, \dots, z_K$.
Let $g_{M,K}$ denote an arbitrary coordinate of the Monte Carlo estimator of either $a_*$ or $S_*$.
In what follows, we denote by $\partial_{z_r}\log w(z)$ the partial derivative w.r.t.~the $r$-th coordinate $z_r$ of a location $z$.

\begin{proposition}\label{pro: BW grad SNR} 
Assume that
\begin{itemize}
    \item[(i)] $\E_{Z\sim q}(w(Z)^4)<\infty$ and $\E_{Z\sim q}(w(Z)^{-16})<\infty$;
    \item[(ii)] $\E_{Z\sim q}\big(w(Z)^s\partial_{z_i}\log w(Z)\big) \ne 0 <\infty$ for all $i$ and $s=2,3$;
    \item[(iii)] $\E_{Z\sim q}\big(w(Z)^s\partial^2_{z_iz_j}\log w(Z)\big) \ne 0 <\infty$, $\E_{Z\sim q}\big(w(Z)^s\partial_{z_i}\log w(Z)\partial_{z_j}\log w(Z)\big) \ne 0 <\infty$ for all $i,j$ and $s=2,3,4$.
\end{itemize} 
Then, we have
\begin{equation}
\text{SNR}\big(g_{M,K}\big) = \Omega(\sqrt{M}) .
\end{equation}
\end{proposition}

The estimator of the affine map \eqref{eq:bw_iw_elbo0} shares the same scaling rate, as it linearly depends on the estimators of $a_*$ and $S_*$.
While this rate is slower than that of the Wasserstein gradient, the BW gradient favorably resolves the SNR degeneracy issue, while maintaining the tractability.
As it eliminates the dependency on $K$ from the rate, the BW-gradient estimator enjoys constant stability across $K$, allowing the use of large $K$ in practice.
The assumptions in the proposition involve the moment properties of the weight function $w(z)$.
While technical, these are typically required in the study of the SNR of the IW-ELBO gradient estimators; see, e.g., \cite{daudel2023alpha,daudel2026importance}.

\subsection{Optimization of IW-ELBO under BW Geometry} \label{sec:algorithm}

A few practical optimization algorithms for the standard ELBO on the BW space have been established in \cite{Lambert2022} and \cite{Diao2023}.
However, relying on the standard ELBO, these frameworks inherit the well-known limitation of underestimating the tail probability.
To address this limitation, we develop an optimization algorithm for the IW-ELBO within the BW geometry.
We call our algorithm the {\it Bures-Wasserstein Importance-Weighted ELBO (BW-IW-ELBO)}. 
Rather than treating the IW-ELBO as a function over the Euclidean space of variational parameters $(m, \Sigma)$, we formulate it as a functional over the BW manifold of Gaussian distributions $\mathcal{N}(m,\Sigma)$.
This reframing naturally lifts the optimization problem from the Euclidean geometry to the Wasserstein geometry, shifting the focus from parameter estimation to distributional optimization.

We formulate the gradient ascent scheme for maximizing the IW-ELBO on the BW space.
On the BW space, the IW-ELBO acts as a functional of $K$ arguments $q_1, \dots, q_K$, while these multiple arguments are all set to one common distribution $q$.
Consider a coordinate descent, where at each iterate we pick a coordinate index $n$ to compute the BW gradient, and update the common distribution $q$.
However, because the BW gradient remains invariant for any index $n$ (see \Cref{rm: invariance}), the choice of the index $n$ does not impact the algorithm.
Therefore, we simply set $n = K$ at every iterate.
The following remark provides an intuitive explanation on why the BW gradient of the IW-ELBO improves the optimization in comparison with the BW gradient of the standard ELBO.

\begin{remark}
    For $K=1$, the self-normalized weight is trivially $1$, meaning $a_*$ and $S_*$ reduce to the simple expectations of the gradient $\nabla \log w(z)$ and Hessian $\nabla^2 \log w(z)$, respectively. 
    For $K > 1$, these terms are modulated by the squared self-normalized importance weight $( w(z_K) / \sum_{i=1}^K w(z_i) )^2$. 
    This factor deprioritizes the gradient and Hessian when a sample $z_K$ yields a small weight $w(z_K) = p(x, z_K) / q(z_K)$ relative to the other $K-1$ samples exploring the domain simultaneously.
    Intuitively, a small weight indicates that the variational family $q$ approximate the target $p$ well at that location.
    This mechanism, therefore, prioritizes the exploration of the domain where $q$ currently underfits $p$, enforcing a mass-covering behavior.
\end{remark}

Let $q^{(k)}=\mathcal{N}(m_k, \Sigma_k)$ denote the iterate at step $k$, initialized at $q^{(0)} = \mathcal{N}(m_0, \Sigma_0)$. 
Adopting the standard discretization scheme for WGF \citep[e.g.][]{Ambrosio2005}, we update the Gaussian density $q^{(k)}$ at step $k$ via the following push-forward operation:
\begin{align}
    q^{(k+1)} = \big( \text{Id} + \eta \nabla^{\mathrm{BW}} \mathcal{F}_K(q^{(k)}) \big)_\# q^{(k)} \label{eq:bw_fb_f}
\end{align}
where $\eta$ denotes the step size of the update.
Since the BW gradient is affine, the update in \eqref{eq:bw_fb_f} corresponds to the following updates for the associated mean and covariance:
\begin{align}
    m_{k+1} = m_{k} + \eta a_* \quad \text{and} \quad \Sigma_{k+1} = (I + \eta S_* ) \Sigma_{k} (I + \eta S_* )  \label{eq:bw_f}
\end{align}
with the characteristic terms $a_*$ and $S_*$ of the BW gradient $\nabla^{\mathrm{BW}} \mathcal{F}_K(q^{(k)})$ at step $k$.
In practice, their Monte Carlo estimates, denoted $\hat{a}_*$ and $\hat{S}_*$, using $M$ independent samples of the $K$ random variables $z_1, \dots, z_K$ are used at each iterate.
The complete procedure is summarized in Algorithm \ref{alg:wgb}.

The covariance update rule in Algorithm \ref{alg:wgb}, given by
\[\Sigma\gets (I + \eta \widehat{S}_* ) \Sigma (I + \eta \widehat{S}_* ),\]
takes a congruence form, guaranteeing that $\Sigma$ remains symmetric and positive semi-definite. 
Furthermore, it remains strictly positive definite provided that $I + \eta \widehat{S}_*$ is invertible. 
To strictly enforce this invertibility despite Monte Carlo noise in practice, we apply an eigenvalue-clipping safeguard to $I + \eta \widehat{S}_*$, bounding its spectrum from below by a positive threshold.
We defer detailed implementation to Appendix \ref{sec:Additional Experimental Details}.

\begin{algorithm}[t]
\caption{BW IW-ELBO VI} \label{alg:wgb}
\KwIn{target density $p$, learning rate $\eta$.}
\KwOut{A Gaussian distribution that approximates the target $p$.}
Initialize the mean and covariance parameters $m, \Sigma$\\
\While {not stop}{
    compute Monte Carlo estimates $\widehat{a}_*$ and $\widehat{S}_*$ from \eqref{eq:bw_iw_elbo1} and \eqref{eq:bw_iw_elbo2}\\
    Update $m\gets m + \eta \widehat{a}_*$\\
    Update $\Sigma\gets (I + \eta \widehat{S}_* ) \Sigma (I + \eta \widehat{S}_* )$.
}
\end{algorithm}

In contrast to the forward-backward approach of \cite{Diao2023} designed for optimizing the ELBO, our approach omits the backward step.
While the backward step is required to ensure convergence due to the non-smoothness of the entropy term in the ELBO, it imposes a substantial computational burden.
Crucially, the IW-ELBO does not admit the decomposition into a convex potential and a non-smooth entropy term exploited in \cite{Diao2023}. 
Moreover, as the smoothness properties of this objective remain ambiguous, the theoretical necessity of the backward step is indeterminate.
We leave the rigorous analysis of this theoretical aspect to future research.

\subsection{Mass-covering Property of BW-IW-ELBO} \label{sec:property}

This section provides a discussion on the mass-covering property of BW-IW-ELBO VI.
Maximizing the ELBO is equivalent to minimizing the Kullback--Leibler divergence
\begin{align}\label{eq:KL div 1}
    \mathrm{KL}(q \,\|\, p) 
    = \int_{\mathcal Z} q(z)\,\log \frac{q(z)}{p(z\mid x)}\,dz,
\end{align}
from the variational approximation \(q(z)\) to the posterior \(p(z\mid x)\), commonly referred to as the reverse KL divergence.  

A key structural property of this divergence is its asymmetry: it penalizes heavily any mass that \(q(z)\) assigns to regions where \(p(z\mid x)\) is negligible or zero, since the integrand becomes unbounded whenever \(q(z) > 0\) but \(p(z\mid x) = 0\).  
As a consequence, any optimizer of \(\mathrm{KL}(q\|p)\) effectively constraints $q(z)$ to place mass only in regions where the posterior density is non-negligible.
This asymmetry has important implications for the behaviour of variational inference.  
In particular, the objective discourages \(q(z)\) from covering low-density regions of \(p(z\mid x)\), even if such regions lie between well-separated modes.  

When \(p(z\mid x)\) is multimodal, this leads to the well-known mode-seeking behaviour of the reverse KL divergence \citep{Li2016,zenn2024differentiable}: the optimal variational approximation often collapses onto a single dominant mode rather than spreading mass across multiple modes.

There is now substantial empirical evidence that the IW-ELBO helps alleviate the mode-seeking behaviour of the standard ELBO \citep{Burda2016,cremer2017reinterpreting,Domke2018,tan2020conditionally,zenn2024differentiable}. While the standard ELBO minimizes the reverse divergence $\mathrm{KL}(q\|p)$ and therefore tends to favour mode-seeking approximations, the IW-ELBO uses multiple importance-weighted samples to define a tighter variational objective. As shown by \cite{Domke2018}, maximizing the IW-ELBO can be interpreted as performing variational inference on an augmented sample space, where the induced variational distribution corresponds to a self-normalized importance sampling approximation.

More precisely, at the inference stage, a sample $z$ from this implicit importance-weighted variational distribution is generated by
\begin{align}
    z_1,\ldots,z_K \sim q(\cdot), \qquad
    z \gets z_i, \qquad\text{with}\quad
    \mathbb{P}(i)=\frac{w(z_i)}{\sum_{k=1}^K w(z_k)} .
\end{align}
This is precisely a sampling-importance-resampling procedure, which is typically more accurate than using samples directly from the standard variational approximation $q$.

Furthermore, \cite{Domke2018} and \cite{maddison2017filtering} show that the IW-ELBO admits the asymptotic expansion
\[
\operatorname{IW\text{-}ELBO}_K(q)
=
\log p(x)
-
\frac{\mathbb{V}_{Z\sim q}(w(Z))}{Kp(x)^2}+o(\frac{1}{K}).
\]
This expansion makes explicit that maximizing the IW-ELBO approximately minimizes the variance of the importance weights. Consequently, variational approximations with insufficient tail coverage are penalized, since light-tailed proposals typically lead to unstable or high-variance importance weights. In this sense, the IW-ELBO promotes a more mass-covering behaviour than the standard ELBO.

The BW-IW-ELBO method combines two complementary ingredients: the mass-covering behaviour induced by the IW-ELBO objective and the intrinsic geometry of the BW manifold. 
The synergy between these two components is crucial, and yields a principled and practically effective variational framework.
While the IW-ELBO objective determines \emph{where} the variational distribution should allocate mass, the BW geometry determines \emph{how} this redistribution should occur in an optimal manner.  
For instance, in the presence of a bimodal posterior, the IW-ELBO objective encourages \(q\) to expand its support to cover both modes.  
The BW geometry then facilitates this expansion through coherent updates of the covariance structure, effectively transporting mass in a way that is consistent with the geometry of the space of distributions. The mass-covering property of the BW-IW-ELBO method is systematically demonstrated in the next section. 

\cite{dhaka2021challenges} study the direct optimization of the {\it inclusive} divergence \(\mathrm{KL}(p\,\|\,q)\), which is well known to encourage mass-covering behaviour. Their experiments show that, particularly in higher dimensions, such objectives may produce posterior approximations with poor mode fidelity and degraded predictive performance. As acknowledged in the literature, however, directly optimizing the inclusive KL objective is itself computationally challenging and can lead to unstable optimization behaviour. Our setting differs in an important way. We do not directly optimize the inclusive divergence but the IW-ELBO objective, which remains rooted in the reverse KL variational framework, while employing the BW optimization geometry to improve optimization efficiency and posterior exploration. Consequently, the improved mass-covering behaviour does not arise from directly enforcing an inclusive-KL objective, but rather from obtaining higher-quality variational approximations within the IW-ELBO framework through a more effective optimization geometry.

\section{Numerical Experiments and Applications}\label{sec:Numerical applications}
We empirically assess the performance of the BW-IW-ELBO method through a series of experiments and applications presented in this section.

Recall that the BW-IW-ELBO method consists of two key components: the IW-ELBO objective and the BW geometry. To disentangle the effects of the objective function and the optimization geometry, we adopt a \(2 \times 2\) factorial design to compare our Gaussian approximation with three baselines, as shown in \Cref{tab:design}.

The first baseline is FB-GVI \citep{Diao2023}, which optimizes the standard ELBO under the BW geometry. Comparing BW-IW-ELBO with FB-GVI isolates the effect of the IW-ELBO objective, as both methods share the same BW geometry.
The second baseline is the Euclidean IW-ELBO, which optimizes the IW-ELBO objective under the standard Euclidean parameterization \((m, L)\), where \(\Sigma = LL^\top\), using the ADAM optimizer. This comparison isolates the effect of the BW geometry.
The third baseline is the Euclidean ELBO, which optimizes the standard ELBO under the same Euclidean parameterization \((m, L)\). This comparison captures the combined gains of the BW-IW-ELBO method arising from both the IW-ELBO objective and the BW geometry.

\begin{table}[h]
\centering
\caption{A $2\times 2$ factorial design to compare various Gaussian approximation methods. The variational approximation is a full-covariance Gaussian $q(z) = \mathcal{N}(m, \Sigma)$.}
\label{tab:design}
\begin{tabular}{lll}
\toprule
 & \textbf{ELBO objective} & \textbf{IW-ELBO objective} \\ \midrule
\textbf{BW geometry} & FB-GVI & BW-IW-ELBO \\
\textbf{Euclidean geometry} & Euclidean ELBO & Euclidean IW-ELBO \\ \bottomrule
\end{tabular}
\end{table}    

\subsection{Simulation Study: Mass-covering Property} \label{sec:Simulation Study: Mass-covering Property}
In this section, we assess the mass-covering property of BW-IW-ELBO on the task of approximating an eggbox distribution \citep{pmlr-v9-murray10a}, a two-dimensional mixture distribution consisting of four equally weighted Gaussian components (``eggs'' $=4$). 

Specifically, we optimize the parameters (mean and covariance) of a single Gaussian distribution $q$ to best approximate this target distribution. The primary objective of this experiment is to evaluate the mass-covering property of the IW-ELBO methods ($K=5$). A Gaussian approximation that distributes its mass across all components of the target can be particularly valuable in downstream applications, such as serving as a proposal distribution for importance sampling or MCMC.

We evaluate the fitted Gaussian moments against the analytical mean $m^\star$ and covariance $\Sigma^\star$ of the target mixture, computed from the component parameters listed in Appendix \ref{sec:Experimental Details for Eggbox}. All methods use the same initialization, $m_0=[6.0,12.0]^\top$ and $\Sigma_0=5I_2$, and are run for $1000$ iterations. The BW covariance-stability safeguard is given in \Cref{sec:Additional Experimental Details}, and additional implementation details are given in Appendix \ref{sec:Experimental Details for Eggbox}.

Figure~\ref{fig:eggbox} shows the evolution of the Gaussian approximations. Under BW geometry, BW-IW-ELBO expands to cover all four modes, whereas FB-GVI concentrates on only part of the target mass. The same pattern appears under the Euclidean parameterization: Euclidean IW-ELBO learns a broad Gaussian spanning the four modes, while Euclidean ELBO collapses onto a single mode.

\begin{figure}
    \centering
    \includegraphics[width=0.9\linewidth]{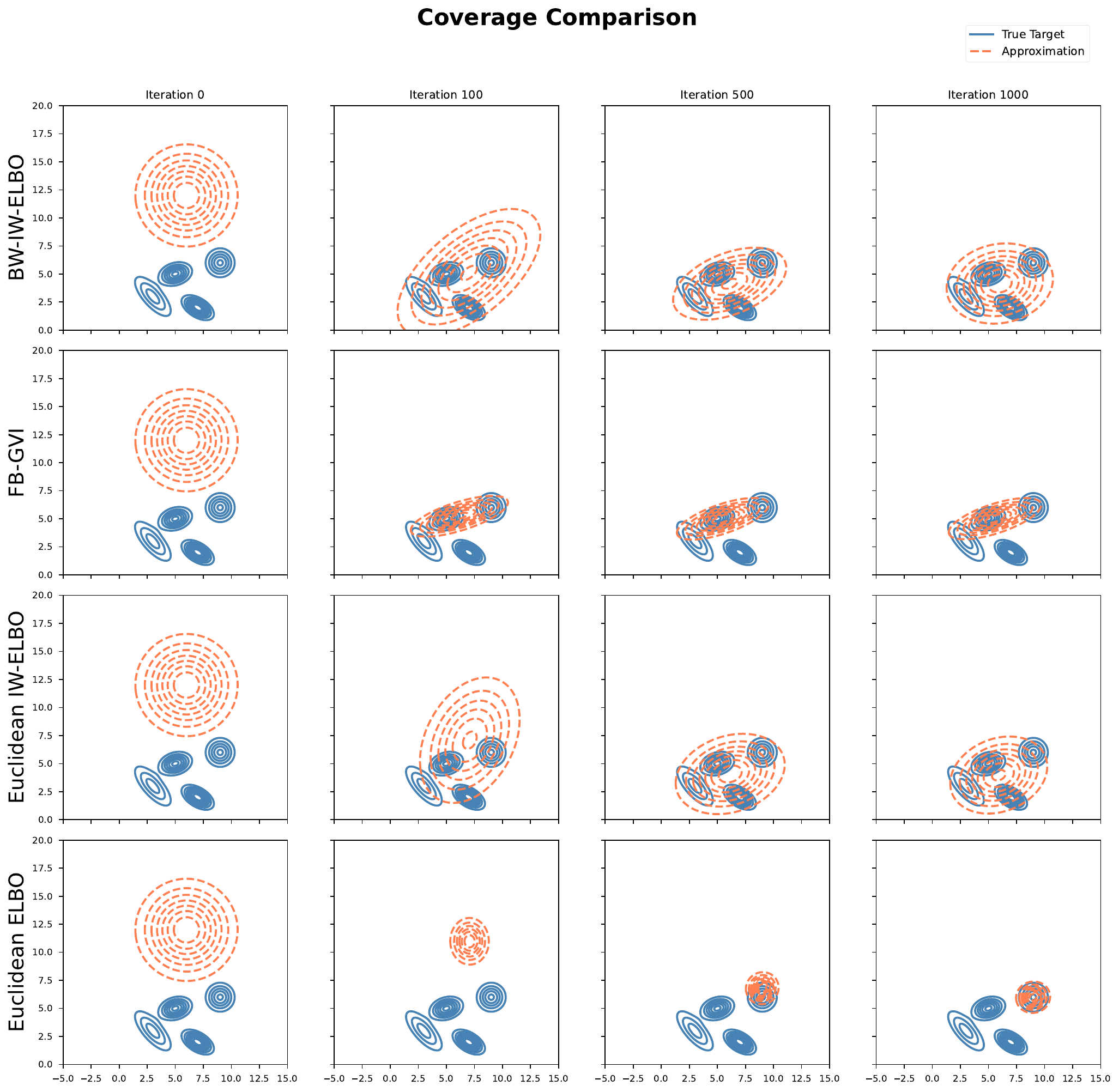}
    \caption{Evolution of the variational approximation for the four-peaked eggbox target.}
    \label{fig:eggbox}
\end{figure}

Table \ref{tab:egg_box_quant} reports how closely the fitted Gaussian approximation matches the analytical moments of the target distribution. We report the squared Euclidean error of the fitted Gaussian mean $\|m - m^\star\|^2$, the squared Frobenius error of the fitted Gaussian covariance $\|\Sigma - \Sigma^\star\|_F^2$,  a Monte Carlo estimate of the
forward KL divergence $\mathrm{KL}(p\,\|\,q) = \mathbb{E}_p[\log p - \log q]$, and the IW-ELBO value (with $K = 5$). With BW geometry fixed, BW-IW-ELBO substantially improves over FB-GVI across all metrics, reducing the forward KL from $3.86$ to $0.78$. The Euclidean comparison shows the same qualitative pattern. These results suggest that, in this experiment, the mass-covering behaviour is mainly due to the IW-ELBO objective rather than the optimization geometry alone.

\begin{table}[H]
\centering
\caption{Performance of the fitted Gaussian approximations on the eggbox target, averaged over $10$ independent replications. Lower values indicate better performance for the first three columns, while higher values are better for the IW-ELBO. Best results within each geometry block are shown in \textbf{bold}.}
\label{tab:egg_box_quant}
\small

\begin{subtable}{\linewidth}
    \centering
    \caption{BW Geometry Methods}
    \begin{tabular}{lcccc}
    \toprule
    Method
    & $\|m-m^\star\|^2\downarrow$
    & $\|\Sigma-\Sigma^\star\|_F^2\downarrow$
    & $\mathrm{KL}(p\,\|\,q)\downarrow$
    & IW-ELBO $\uparrow$ \\
    \midrule
    BW-IW-ELBO        & $\mathbf{0.02\pm 0.03}$ & $\mathbf{1.64\pm 1.29}$ & $\mathbf{0.78\pm0.05}$ & $\mathbf{-0.26\pm 0.03}$\\
    FB-GVI            & $1.24\pm 0.03$ & $7.06\pm 1.09$          & $3.86\pm0.27$ & $-0.54\pm 0.03$\\
    \bottomrule
    \end{tabular}
\end{subtable}

\vspace{1em} 

\begin{subtable}{\linewidth}
    \centering
    \caption{Euclidean Geometry Methods}
    \begin{tabular}{lcccc}
    \toprule
    Method
    & $\|m-m^\star\|^2\downarrow$
    & $\|\Sigma-\Sigma^\star\|_F^2\downarrow$
    & $\mathrm{KL}(p\,\|\,q)\downarrow$
    & IW-ELBO $\uparrow$ \\
    \midrule
    Euclidean IW-ELBO & $\mathbf{0.02\pm 0.02}$ & $\mathbf{2.12\pm 2.41}$          & $\mathbf{0.77\pm0.03}$ & $\mathbf{-0.25\pm 0.04}$\\
    Euclidean ELBO    & $12.85 \pm 0.05$ & $35.83 \pm 0.09$      & $17.80 \pm 0.34$ & $-1.37\pm 0.02$\\
    \bottomrule
    \end{tabular}
\end{subtable}
\end{table}

\subsection{Simulation Study: Performance of Gradient Estimators}
\label{sec:Simulation Study: Performance of gradient estimators}
This section has two goals. First, using a Bayesian logistic regression model, we numerically evaluate the signal-to-noise ratio (SNR) of the Wasserstein gradient estimator, the BW gradient estimator, and the Euclidean gradient estimator of the IW-ELBO. Second, we assess the optimization efficiency of the BW and Euclidean gradient estimators by comparing the convergence behaviour of the BW-IW-ELBO and Euclidean IW-ELBO methods.

The Bayesian logistic regression model assumes a Bernoulli distribution for the binary outcome $y_i \in \{0, 1\}$ given a feature vector $X_i$:
\[
p(y_i \mid X_i, z) = \mathrm{Bernoulli}(\sigma(X_i^\top z)),
\]
where $\sigma(\cdot)$ is the sigmoid function and $z$ is the coefficient vector. Given data $\mathcal{D} = \{(X_i, y_i)\}_{i=1}^n$, the inference task is to approximate the posterior $p(z \mid \mathcal{D})$ by a Gaussian distribution with mean $m$ and full-covariance matrix $\Sigma$.

\textbf{SNR of the Gradient Estimators}

We empirically study the SNR scaling of the gradient estimators in Section \ref{eq:Euclidean vs Wasserstein} on a synthetic Bayesian logistic regression problem with prior $z\sim\mathcal N(0,I_d)$ and fixed variational distribution $q=\mathcal N(0,I_d)$. We vary the number of importance samples $K$, the number of Monte Carlo replicates $M$, and the latent dimension $d\in \{20,50,80\}$, rescaling the design matrix by $\sqrt{c/d}$ ($c$ is a constant) to keep the logit scale comparable across dimensions (Appendix \ref{sec:Experimental Details for BLR experiment}).

We compare the Wasserstein, BW, and Euclidean reparameterized gradient estimators, whose predicted SNR rates are respectively $\Omega(\sqrt{MK})$, $\Omega(\sqrt M)$, and $O(\sqrt M/\sqrt K)$ according to Theorem~\ref{the: Wass grad SNR}, Proposition~\ref{pro: BW grad SNR}, and \citet{rainforth2018tighter}. Since these results provide asymptotic one-sided bounds rather than exact rates, our objective is not to verify the theoretical exponents quantitatively, but rather to assess whether the empirical log--log slopes exhibit the predicted qualitative behavior: positive for Wasserstein (increasing with $K$), approximately zero for BW (essentially independent of $K$), and negative for Euclidean (decreasing with $K$).

Because the three estimators are defined under different geometries, their SNR magnitudes are not directly comparable. We therefore compare their scaling through the slopes of $\log\mathrm{SNR}$ against $\log K$ and $\log M$. We vary $K\in\{10,100,200,500,1{,}000,2{,}000,4{,}000,8{,}000,10{,}000\}$ at fixed $M=1$, and vary $M\in\{1,2,4,8,16\}$ at fixed $K=100$, for each $d\in\{20,50,80\}$. The SNR is computed coordinatewise and then averaged across coordinates, with the BW estimator using the combined coordinates of $(\hat a_*,\operatorname{diag}\hat S_*)$. Further simulation details are given in Appendix \ref{sec:Experimental Details for BLR experiment}.

Figure \ref{fig:snr_scaling} and Table \ref{tab:snr_scaling} summarize the results. For $K$-scaling at fixed $M=1$, the Wasserstein SNR increases, the BW SNR remains nearly flat, and the Euclidean SNR decreases, with little dependence on $d$. The fitted $\log K$ slopes are $0.53$-$0.56$ for Wasserstein, $-0.01$-$-0.03$ for BW, and around $-0.35$ for Euclidean. These empirical results are typically consistent with the theory. 
For $M$-scaling at fixed $K=100$, all fitted slopes lie in $[0.48,0.53]$, confirming the common $\sqrt M$ dependence across estimators and dimensions.

\begin{figure}[t]
  \centering
  \begin{subfigure}{0.48\textwidth}
    \centering
    \includegraphics[width=\linewidth]{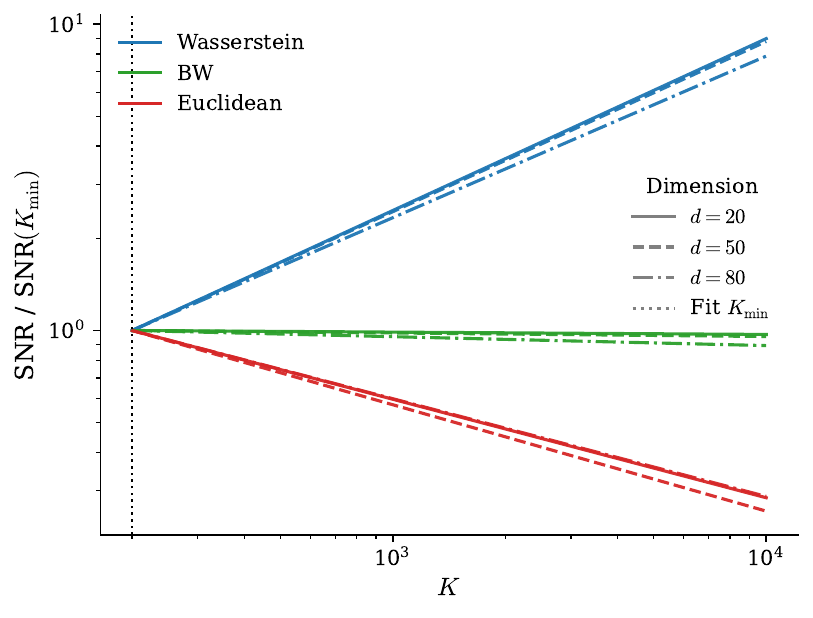}
    \caption{$K$-scaling at fixed $M = 1$.}
    \label{fig:snr_K}
  \end{subfigure}
  \hfill
  \begin{subfigure}{0.48\textwidth}
    \centering
    \includegraphics[width=\linewidth]{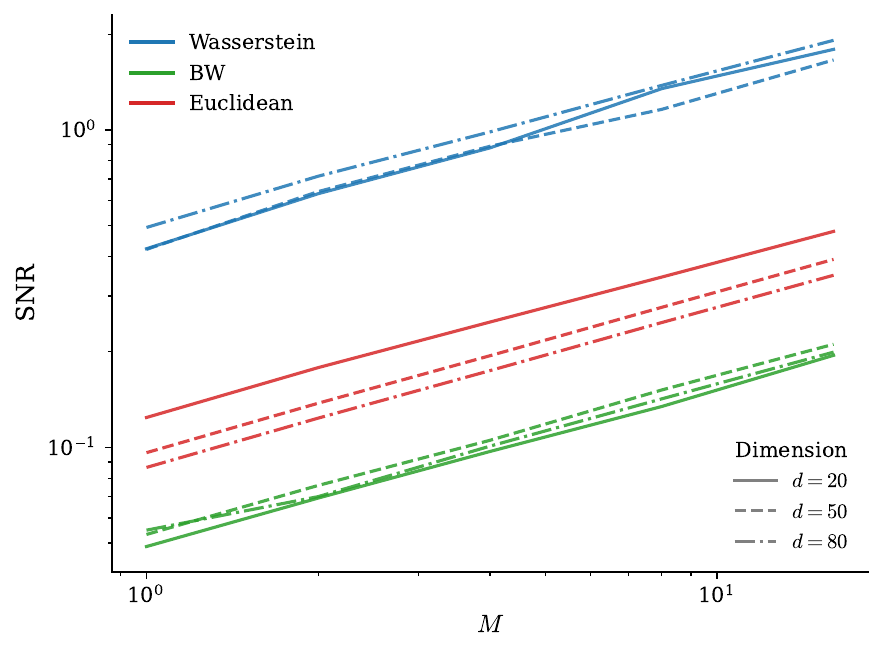}
    \caption{$M$-scaling at fixed $K = 100$.}
    \label{fig:snr_M}
  \end{subfigure}
  \caption{SNR scaling of the three gradient estimators across latent dimensions $d \in \{20, 50, 80\}$ on the Bayesian logistic regression model. Colour denotes the estimator and linestyle denotes the latent dimension. \textbf{(a)} SNR versus the number of importance samples $K$ at fixed $M = 1$; the dotted vertical line marks $K_{\min}=200$, the lower edge of the slope-fitting window. \textbf{(b)} SNR versus the number of Monte Carlo replicates $M$ at fixed $K = 100$.}
  \label{fig:snr_scaling}
\end{figure}

\begin{table}[t]
  \centering
  \caption{Fitted log-SNR slopes across latent dimension $d$, with leave-one-seed-out jackknife standard errors. The $\log K$ slopes are
 fitted at $M = 1$ over $K \ge 200$, and the $\log M$ slopes are fitted at
 $K = 100$. The qualitative predictions are that the Wasserstein SNR
 increases with $K$, the BW SNR is essentially independent of $K$, and the
 Euclidean SNR decreases with $K$, while all three estimators scale as
 $\sqrt M$ in the number of Monte Carlo replicates.}
  \label{tab:snr_scaling}
  \begin{tabular}{c ccc ccc}
    \toprule
    & \multicolumn{3}{c}{$\log K$ slope ($M = 1$)}
    & \multicolumn{3}{c}{$\log M$ slope ($K = 100$)} \\
    \cmidrule(lr){2-4}\cmidrule(lr){5-7}
    $d$ & Wasserstein & BW & Euclidean
        & Wasserstein & BW & Euclidean \\
    \midrule
    $20$ & $+0.56 \pm 0.01$ & $-0.01 \pm 0.03$ & $-0.32 \pm 0.05$
         & $+0.53 \pm 0.03$ & $+0.50 \pm 0.02$ & $+0.48 \pm 0.01$ \\
    $50$ & $+0.56 \pm 0.01$ & $-0.01 \pm 0.03$ & $-0.35 \pm 0.03$
         & $+0.48 \pm 0.03$ & $+0.50 \pm 0.02$ & $+0.51 \pm 0.01$ \\
    $80$ & $+0.53 \pm 0.01$ & $-0.03 \pm 0.05$ & $-0.32 \pm 0.03$
         & $+0.49 \pm 0.01$ & $+0.48 \pm 0.02$ & $+0.50 \pm 0.02$ \\
    \bottomrule
  \end{tabular}
\end{table}

\textbf{Convergence of BW-IW-ELBO and Euclidean IW-ELBO}

We next compare BW-IW-ELBO and Euclidean IW-ELBO, using various $K\in\{10,50,100,200\}$, for approximating the posterior in the Bayesian logistic regression model. Both methods are initialized at $m_0 = 0$, $\Sigma_0 = 5I$, use $M = 100$ $K$-tuples per iteration, and are run over $10$ evaluation seeds. 
Implementation details including step sizes and convergence stopping rule are given in \Cref{sec:Experimental Details for BLR experiment}.

Table~\ref{tab:time-to-target} and Figure~\ref{fig:boxplots}(a) show that BW-IW-ELBO consistently converges in fewer iterations than the Euclidean baseline across all tested values of $K$. The median iteration count is $22$ for BW-IW-ELBO versus $115$ for the Euclidean IW-ELBO at $K = 10$, and $70$ versus $578$ at $K = 50$. At $K = 200$ the BW method remains faster in the median ($1{,}204$ versus $2{,}009$ iterations), although its convergence times are more variable across seeds, with interquartile range $[530, 1{,}453]$. This variability is associated with the large selected BW step ($\eta = 4.0$).

The same pattern is observed in wall-clock time, as shown in Table~\ref{tab:time-to-target} and Figure~\ref{fig:boxplots}(b). Compared with the Euclidean baseline, BW-IW-ELBO reaches convergence faster for every tested value of \(K\), with median speed-ups ranging from \(2.6\times\) at \(K=200\) (\(88.6\) vs.\ \(226.2\) seconds) to \(12\times\) at \(K=50\) (\(4.1\) vs.\ \(49.5\) seconds). For the intermediate values \(K=10\) and \(K=100\), the median speed-ups are \(6.7\times\) and \(3.7\times\), respectively.
BW-IW-ELBO is also computationally cheaper on a per-iteration basis (see Table~\ref{tab:per_iter_cost} in the Appendix), primarily because it avoids backpropagation through the full \(M\times K\) computation graph required by the Euclidean reparameterization gradient estimator.

\begin{figure}[htbp]
    \centering
    \begin{subfigure}[b]{0.48\textwidth}
        \centering
        \includegraphics[width=\textwidth]{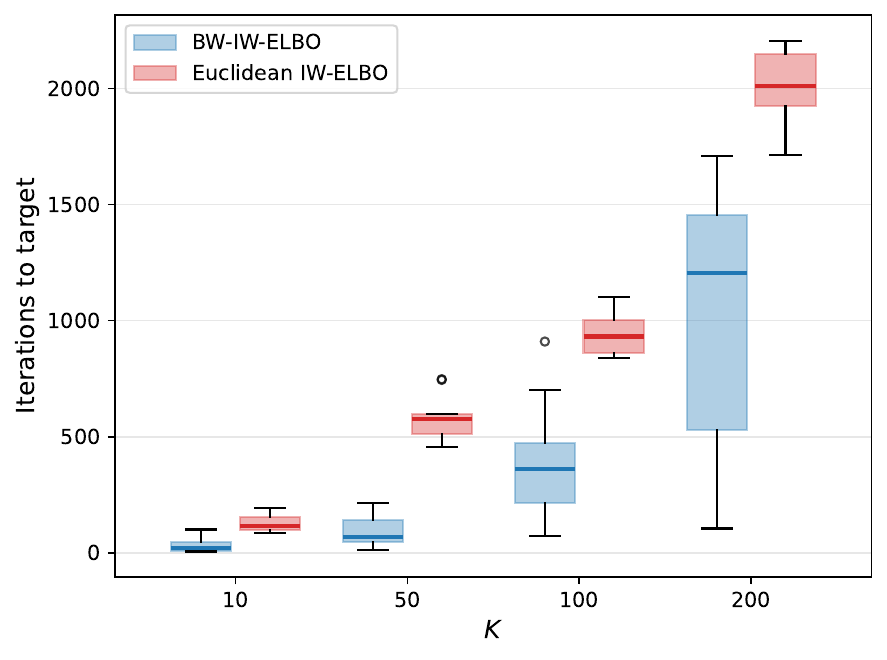}
        \caption{Iterations to Target}
    \end{subfigure}
    \hfill 
    \begin{subfigure}[b]{0.48\textwidth}
        \centering
        \includegraphics[width=\textwidth]{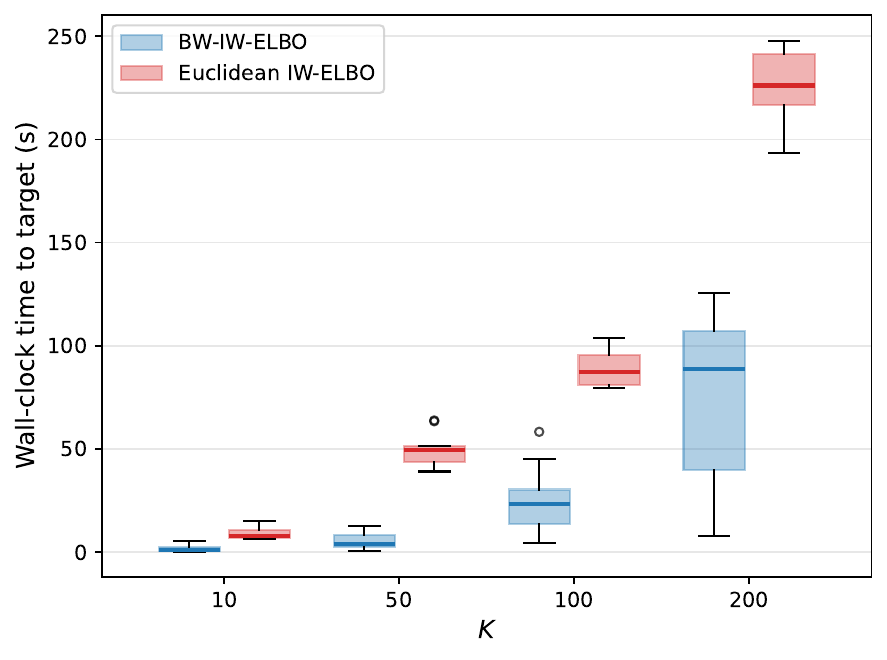}
        \caption{Wall-clock Time to Target}
    \end{subfigure}    
    \caption{Iterations and wall-clock time to the common threshold $-473.72$, shown as box plots over $10$ random seeds for each $K$. Wall-clock time excludes evaluation overhead. Per-iteration costs and implementation details are reported in Table \ref{tab:per_iter_cost} in \cref{sec:Experimental Details for BLR experiment}.}
    \label{fig:boxplots}
\end{figure}

\begin{table}[htbp]
    \centering
    \caption{The table reports the median, first quartile (Q1), and third quartile (Q3) of the number of iterations and CPU time (in seconds) until convergence, computed over 10 random seeds. The last column reports the median final standard ELBO, averaged over the last 100 evaluations.}
    \begin{tabular}{cl|ccc}
        \toprule
        $K$ & Method & iter med. [Q1,Q3] & time med. [Q1,Q3] & final ELBO med. \\
        \midrule
        10  & BW-IW-ELBO        & 22 [10, 49]     & 1.2 [0.5, 2.5]     & -472.72 \\
          & Euclidean IW-ELBO & 115 [99, 156]   & 8.0 [7.0, 10.5]    & -473.11 \\
        \midrule
        50  & BW-IW-ELBO        & 70 [46, 141]    & 4.1 [2.7, 8.3]     & -472.72 \\
          & Euclidean IW-ELBO & 578 [514, 599]  & 49.5 [43.9, 51.5]  & -473.14 \\
        \midrule
        100 & BW-IW-ELBO        & 362 [216, 474]  & 23.5 [13.8, 30.4]  & -472.72 \\
         & Euclidean IW-ELBO & 932 [861, 1004] & 87.5 [80.9, 95.4]  & -473.26 \\
        \midrule
        200 & BW-IW-ELBO        & 1204 [530, 1453] & 88.6 [39.7, 107.1]  & -472.72 \\
         & Euclidean IW-ELBO & 2009 [1924, 2147] & 226.2 [216.9, 241.4] & -473.12 \\
        \bottomrule
    \end{tabular}
    \label{tab:time-to-target}
\end{table}

Figure \ref{fig:elbo_tracking} shows the fixed-seed standard ELBO trajectories near the optimum. BW-IW-ELBO reaches the plateau earlier and achieves a higher standard ELBO value than the Euclidean baseline for every tested $K$, consistent with Table~\ref{tab:time-to-target}.

\begin{figure}
    \centering
    \includegraphics[width=1.0\linewidth]{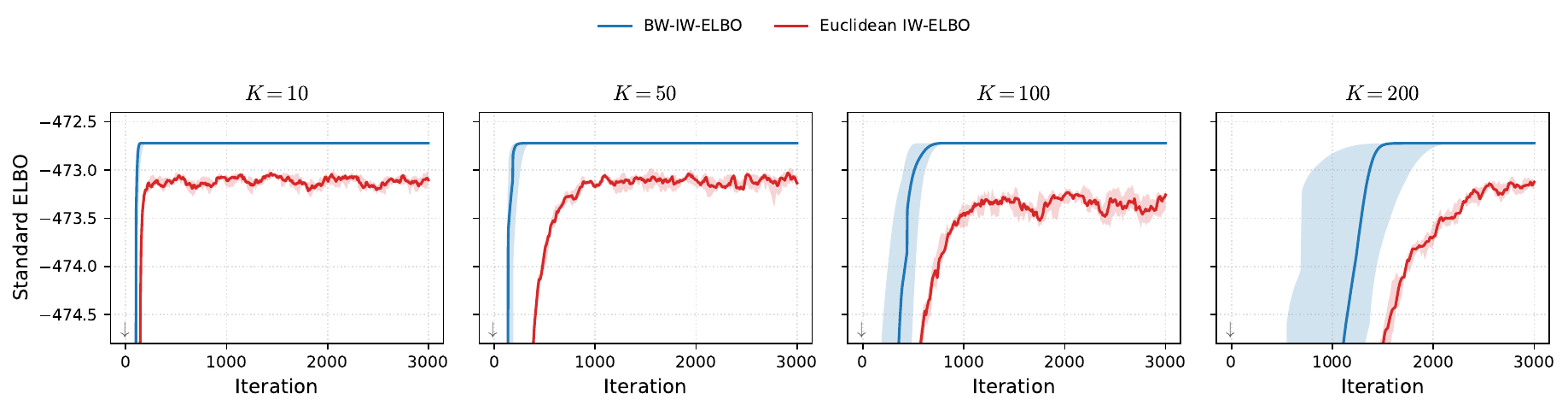}
    \caption{Standard ELBO tracking near the optimum for each $K$. Curves show the pointwise median over $10$ seeds of trajectories smoothed with a window-$100$ moving average; shaded bands show the interquartile range.}
    \label{fig:elbo_tracking}
\end{figure}

\subsection{Application: Census Binary Data}
\label{sec:Application: Census Binary Data}
Consider the Census dataset from the UCI Machine Learning Repository, where the task is to predict whether an individual's income exceeds \$50{,}000 based on their demographic and socio-economic attributes. After correlation-based screening, we use eight covariates and an intercept, giving a $d=9$ with $N=45{,}221$ observations. We compare all four methods in Table~\ref{tab:design} against an MCMC posterior reference using Bayesian logistic regression; implementation details are given in \Cref{sec:Empirical Details for census Data}.

Figure~\ref{fig:postcontour} visually compares the posterior distributions, estimated by BW-IW-ELBO, Euclidean IW-ELBO and MCMC. The two variational inference methods fix the Gaussian variational family and IW-ELBO objective while varying only the optimization geometry. The marginal densities and pairwise contours show that the BW-IW-ELBO estimates closely match the MCMC posteriors, whereas Euclidean IW-ELBO exhibits substantial discrepancies. Further comparisons across all four variational inference methods are provided in Figure \ref{fig:marginal_post} in the Appendix.

Table~\ref{tab:blr_results} provides a quantitative comparison of the variational methods listed in Table~\ref{tab:design}. We report the final standard ELBO and the normalized effective sample size (nESS). Given samples $z_1,\ldots,z_R$ from the fitted Gaussian approximation and importance weights $w_i=\frac{p(\mathcal{D},z_i)}{q(z_i)}$, the effective sample size is
$$
\mathrm{ESS}
=
\frac{\big(\sum_{i=1}^R w_i\big)^2}{\sum_{i=1}^R w_i^2}.$$
We report $\mathrm{nESS}=\mathrm{ESS}/R$ as a percentage, with larger values indicating a more effective importance proposal.

Holding the geometry fixed to BW, replacing the standard ELBO with the IW-ELBO (i.e., comparing FB-GVI with BW-IW-ELBO) increases the nESS from $86.5\%$ to $99.8\%$ and improves the ELBO from $-15754.12$ to $-15753.98$. Holding the objective fixed to IW-ELBO, switching from BW to Euclidean geometry reduces the nESS dramatically from $99.8\%$ to $10.5\%$, while also worsening the ELBO from $-15753.98$ to $-15758.61$.

At first glance, the two Euclidean rows of Table~\ref{tab:blr_results} appear counterintuitive: under Euclidean geometry, optimizing the tighter IW-ELBO yields a worse Gaussian approximation than optimizing the standard ELBO. This phenomenon can be explained by the SNR pathology identified by \citet{rainforth2018tighter}: as the number of importance samples \(K\) increases, the SNR of the Euclidean reparameterization gradient deteriorates, hindering effective optimization of the IW-ELBO objective. In contrast, Proposition~\ref{pro: BW grad SNR} shows that the BW gradient estimator maintains an SNR that does not deteriorate with \(K\), explaining why BW-IW-ELBO retains the expected advantage of the tighter IW-ELBO objective in this experiment.

\begin{table}[htbp]
\centering
\caption{Comparison of variational methods with $d=9$ on the Census dataset. We report the final ELBO and normalized effective sample size (nESS). All entries are averaged over 10 independent training seeds; uncertainties are seed-to-seed standard deviations.}
\label{tab:blr_results}
\begin{tabular}{lcc}
\toprule
\textbf{Method} & \textbf{ELBO} 
& \textbf{nESS} \\
\midrule
BW-IW-ELBO      & $\mathbf{-15753.98 \pm 0.00}$ 
& $\mathbf{99.8\% \pm 0.0\%}$ \\
FB-GVI          & $-15754.12 \pm 0.02$
& $86.5\% \pm 0.7\%$ \\
Euclidean IW-ELBO & $-15758.61 \pm 0.43$ 
& $10.5\% \pm1.6\%$ \\
Euclidean ELBO  & $-15754.42 \pm 0.06$ 
& $60.3\% \pm 3.2\%$ \\
\bottomrule
\end{tabular}
\end{table}

 \begin{figure}[h!]
    \centering
    \includegraphics[width=\textwidth]{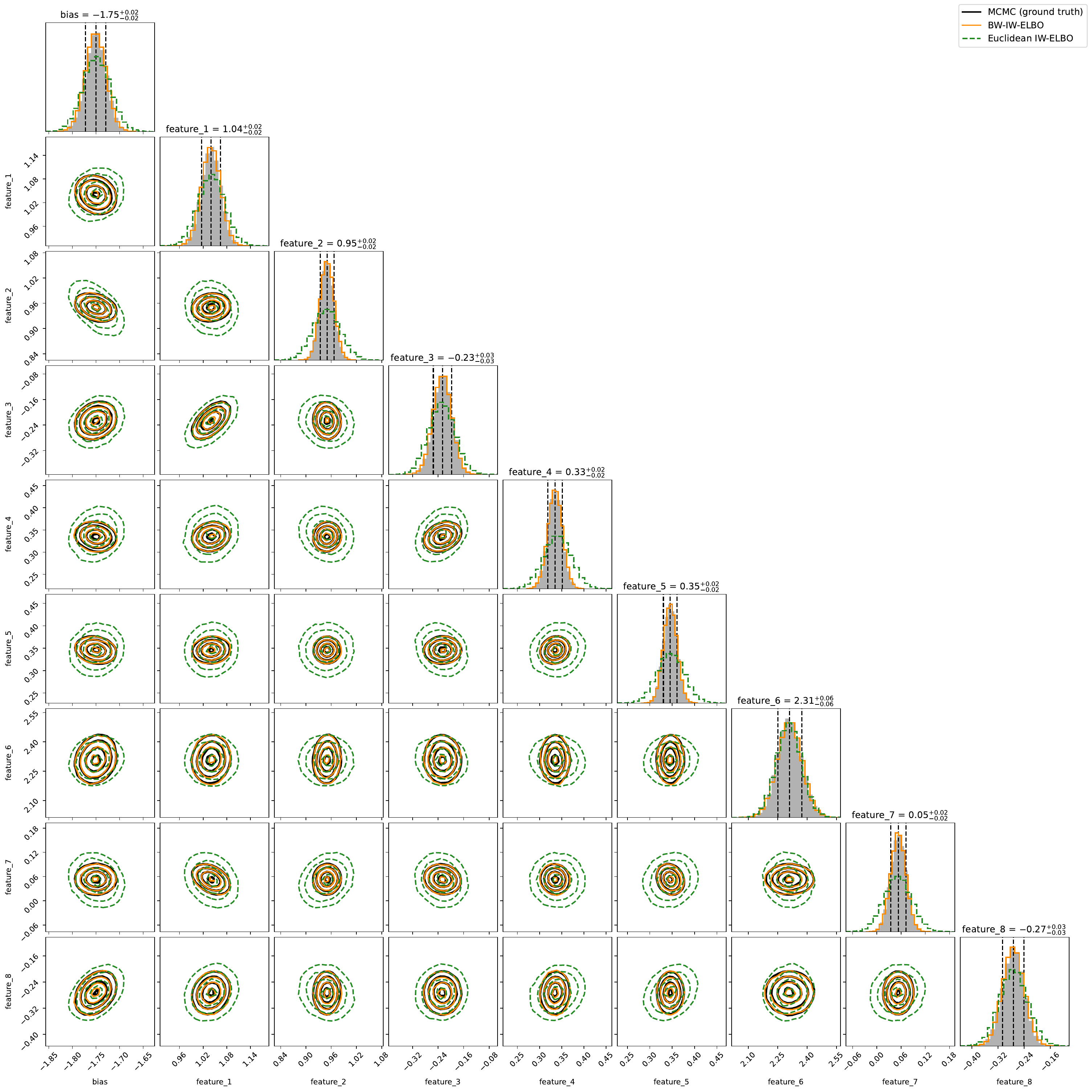}
    \caption{Pairwise-dimensional and marginal posteriors, estimated by the three methods: MCMC, BW-IW-ELBO and Euclidean IW-ELBO}
    \label{fig:postcontour}
\end{figure}

\textbf{Scaling to higher dimensions}

We repeat the Census experiment using the full feature set. After removing exact rank deficiencies by column-pivoted QR (Appendix \ref{sec:Empirical Details for census Data}), the design matrix has $d=96$ columns and a condition number of $8.0$. Since the posterior is sharply concentrated, with marginal standard deviations in $[0.0123,0.0645]$, we use a Laplace approximation as the reference posterior.

Table~\ref{tab:census-highdim} shows that the qualitative pattern persists at $d=96$. Both BW methods substantially outperform the Euclidean baselines, with much higher nESS (about $55$-$65\%$ versus essentially zero) and much smaller KL divergence relative to the Laplace reference. BW-IW-ELBO and FB-GVI are comparable in this near-Gaussian setting, suggesting that the dominant improvement here comes from the Bures-\-Wasserstein geometry. The complementary benefit of IW-ELBO is seen more clearly in the multimodal eggbox experiment (\Cref{sec:Simulation Study: Mass-covering Property}).

\begin{table}[ht]
\centering
\caption{Comparison of variational methods with $d=96$ on the Census dataset. We report the final ELBO, normalized effective sample size (nESS), and Gaussian divergence $\mathrm{KL}(p\,\|\,q)$, where $p$ is the Laplace reference and $q$ is the fitted Gaussian approximation. Entries are means over 10 training seeds, with seed-to-seed standard deviations.}
\label{tab:census-highdim}
\small
\setlength{\tabcolsep}{4pt}
\begin{tabular}{lccc}
\toprule
\textbf{Method} & \textbf{ELBO} $\uparrow$ 
& \textbf{nESS} $\uparrow$
& \textbf{KL($\mathbf{p\,\|\,q}$)} $\downarrow$ \\
\midrule
BW-IW-ELBO
  & $-15386.13 \pm 0.01$
  & $64.77\% \pm 7.7\%$
  & $0.38 \pm 0.01$ \\
FB-GVI
  & $-15386.11 \pm 0.00$
  & $55.04\% \pm 12.5\%$
  & $0.36 \pm 0.01$ \\
Euclidean IW-ELBO
  & $-15897.05 \pm 5.39$
  & $0.00\% \pm 0.0\%$
  & $75.59 \pm 2.34$ \\
Euclidean ELBO
  & $-15486.72 \pm 1.62$
  & $0.00\% \pm 0.1\%$
  & $32.32 \pm 0.83$ \\
\bottomrule
\end{tabular}
\end{table}

\section{Extension to Bures-Wasserstein VR-IWAE} \label{sec:VR-IWAE}

This section explores the connections between the BW-IW-ELBO framework and the VR-IWAE bound.
The VR-IWAE bound represents another emerging objective function for VI, sharing key properties with the IW-ELBO.
The majority of the theoretical analysis presented in \Cref{eq:BW-IW-ELBO} readily extend to the VR-IWAE bound.
Since the primary focus of this work is investigating the BW-IW-ELBO, our empirical evaluation for the VR-IWAE bound is restricted to a proof-of-concept.
Yet, it offers valuable insights into the potential for generalizing the BW-IW-ELBO in sequel.

\subsection{Variational Rényi Importance-Weighted Autoencoder Bound}

First, we briefly review the formulation of the VR-IWAE bound proposed in 
\cite{daudel2023alpha}. 
It combines the Variational Rényi (VR) bound of \cite{Li2016} with the IW-ELBO bound.
Given a power hyperparameter $\alpha\in[0,1)$, the VR-IWAE bound is defined as
\begin{align}
	\operatorname{VR-IWAE}^{(\alpha)}_K(q_1, \dots, q_K) = \frac{1}{1-\alpha}\mathop{\E}_{z_1\sim q_1, \dots, z_K \sim q_K} \left[ \log \left( \frac{1}{K} \sum_{i=1}^{K} \frac{p(x,z_i)^{1-\alpha}}{q_i(z_i)^{1-\alpha}} \right) \right] ,
\end{align}
where each $q_i$ is set to a common variational distribution $q$.
When $\alpha=0$, the VR-IWAE bound coincides with the IW-ELBO. 
When $\alpha>0$, it introduces a highly flexible bias-variance trade-off, improving the efficiency of VI in complex applications.

In standard parametric VI, the variational distribution $q$ is chosen from a parametric family $q_\psi$ parameterized by a vector $\psi$.
Unlike the VR bound, a major advantage of the VR-IWAE bound is that it admits an unbiased Euclidean gradient estimator with respect to $\psi$ \citep{daudel2023alpha}.
Furthermore, the SNR of the gradient estimator scales as $O(\sqrt{K})$, unlike that of the IW-ELBO.

\subsection{Wasserstein Gradient of the VR-IWAE Bound}

In what follows, we extend our analysis of the Wasserstein gradient of the IW-ELBO to the VR-IWAE bound.
As in \Cref{sec:WG_IW_ELBO}, we began by viewing the VR-IWAE bound as a functional over $K$ probability distributions $q_1, \dots, q_K$.
For an arbitrary index $n$, we fix the $K-1$ arguments $\{ q_i \}_{i \ne n}$ to the common distribution $q$ and define a functional $\G_n(q_n)$ on $\mathbb{W}_2(\R^d)$ by
\begin{equation}
	\G_n(q_n) := \frac{1}{1-\alpha} \mathop{\E}_{z_n \sim q_n} \left[\mathop{\E}_{ \{ z_1, \dots, z_K \} \setminus z_n \overset{i.i.d.}{\sim} q } \log \left( \frac{1}{K} \sum_{i = 1}^{K} \frac{p(x,z_i)^{1-\alpha}}{q(z_i)^{1-\alpha}} \right) \right] ,
\end{equation}
This expression represents the VR-IWAE bound viewed as a functional of the $n$-th argument $q_n$.
As in \Cref{sec:WG_IW_ELBO}, our aim is to derive the coordinatewise Wasserstein gradient of the VR-IWAE bound with respect to $q_n$, demonstrating that the form of this gradient remains invariant to the index $n$.

\begin{proposition}\label{pro: proposition VR-IWAE}
Suppose standard conditions that allow interchanging derivative and expectation. 
The coordinatewise Wasserstein gradient of the VR-IWAE bound $\G_n(q_n)$ at $q_n = q$ is given by
\begin{multline}\label{eq:Wass grad VR-IWAE}
	\nabla^W \left[ \G_n(q) \right](z_n) = \\
    \mathop{\E}_{ \{ z_1, \dots, z_K \} \setminus z_n \overset{i.i.d.}{\sim} q } \lrsb{\lr{\alpha\frac{ w(z_n)^{1-\alpha} }{ \sum_{i=1}^{K} w(z_i)^{1-\alpha}} + (1-\alpha) \bigg(\frac{  w(z_n)^{1-\alpha}}{ \sum_{i=1}^{K} w(z_i)^{1-\alpha} } \bigg)^2} \nabla_{z_n} \log w(z_n)} \in \R^d .
\end{multline}
\end{proposition}

The proof is deferred to the Appendix.
By exactly the same argument in \Cref{sec:WG_IW_ELBO}, this expression for the coordinatewise Wasserstein gradient \eqref{eq:Wass grad VR-IWAE} is invariant with respect to the index $n$.
Therefore, without loss of generality, we can focus on the case $n = K$ in what follows.

The coordinatewise Wasserstein gradient \eqref{eq:Wass grad VR-IWAE} shares a similar form with the doubly-reparameterized gradient of the parametric VR-IWAE bound.
\cite{daudel2024learning} showed that the SNR of the doubly-reparameterized gradient estimator scales as $O(\sqrt{K})$.
A Monte Carlo estimator of the Wasserstein gradient \eqref{eq:Wass grad VR-IWAE} for the VR-IWAE bound exhibits the same SNR scaling rate.
Let $g^{r}_{\alpha,M,K}(z)$ denote the $r$-th coordinate of the Monte Carlo estimator of the Wasserstein gradient \eqref{eq:Wass grad VR-IWAE} evaluated $z$, using $M$ independent samples of the $K-1$ variables $z_1, \dots, z_{K-1}$. 

\begin{theorem}\label{the: SNR of VR-IWAE} Let $\alpha\in[0,1)$ and assume that $\E_{Z\sim q}(w(Z)^{4(1-\alpha)})<\infty$ and $\E_{Z\sim q}(w(Z)^{-12(1-\alpha)})<\infty$. Then, for any $z_n$ with $w(z_n)>0$,
\begin{equation}
\text{SNR}\big(g^r_{\alpha,M,K}(z_n)\big) = \Omega(\sqrt{MK}),\;\;r=1,...,d.
\end{equation}
\end{theorem}
The proof is essentially same as that of Theorem \ref{the: Wass grad SNR} and hence omitted.

\subsection{Optimization of the VR-IWAE Bound on the BW Space}

To formulate the gradient descent scheme for maximizing the VR-IWAE bound, we first derive the BW gradient of the VR-IWAE bound.
For notational convenience, let $G(z)$ denote the coordinatewise Wasserstein gradient \eqref{eq:Wass grad VR-IWAE} evaluated at $z$ with $n = K$.
Following the same derivation as in \Cref{sec:algorithm}, the BW gradient of the VR-IWAE bound $\G_K(q)$ at a Gaussian density $q = \mathcal{N}(m,\Sigma)$ is given by 
\begin{align}
    \nabla^\mathrm{BW}[ \G_K(q) ](z ) = a_* + S_* (z - m) ,
\end{align}
where the vector $a_*$ and matrix $S_*$ are defined by, respectively,
\begin{align}
    a_*  = \E_{z_k \sim q}[ G(z_K) ], \quad \text{and} \quad S_*  = \E_{z_K \sim q}[ \nabla G(z_K) ] . \label{eq:BW_terms_VR_IWAE}
\end{align}
The explicit expressions for $a_*$ and $S_*$, along with their Monte Carlo estimates, can be found in \Cref{sec:BW_derivation_VE_IWAE}.
Applying essentially the same proof as in \Cref{pro: BW grad SNR}, we can show that the SNR of this BW gradient scales as $\Omega(1)$.
Algorithm \ref{alg:wgb}, presented in \Cref{sec:algorithm}, is readily adapted to this BW gradient by substituting the newly derived expressions for the vector $a_*$ and matrix $S_*$.

We evaluate the efficacy of this VI algorithm using a unimodal but challenging target, commonly known as the banana-shaped distribution \citep{Heikki2001}.
Let $q$ denote the density function of a $d$-dimensional Gaussian distribution $\mathcal{N}(0,\Sigma)$ with covariance matrix $\Sigma=\mathrm{diag}(100,1,\dots,1)$.
For a given hyperparameter $b > 0$, the target density is defined as $\pi(x) = q(\phi(x))$, where
\begin{align}
    \phi(x)=(x_1, x_2 + b x_1^2 - 100 b, x_3, ..., x_d) \in \R^d .
\end{align}
Despite its unimodality, this target poses a significant challenge for VI methods due to its highly non-Gaussian geometry characterized by extended, narrow tails. Varying the hyperparameter $b$ adjusts the curvature and, consequently, the degree of non-Gaussianity \citep{Heikki2001}; we set $b=0.03$ in our experiments. 
We approximate this target with a single multivariate Gaussian, comparing the four algorithms of our factorial design together with the BW-VR-IWAE method at $\alpha = 0.1$ and $\alpha = 0.9$, all initialized identically. Figure \ref{fig:banana} overlays the learned Gaussian contours on samples from the true target and reports the estimated marginal variances against the true values $\text{Var}(x_1) = 100$ and $\text{Var}(x_2) = 19$. The bent coordinate $x_2$ is the discriminating one: methods that collapse onto the high-density ridge recover its variance poorly. The standard-ELBO methods collapse most strongly - Euclidean ELBO underestimates both directions ($\widehat{\text{Var}}(x_1) = 14.3$, $\widehat{\text{Var}}(x_2) = 1.0$), and FB-GVI captures $x_1$ ($87.8$) but collapses $x_2$ to $1.0$. Switching to the IW-ELBO objective markedly improves the $x_2$ dispersion (Euclidean IW-ELBO: $7.1$; BW-IW-ELBO: $10.9$), and within the IW-ELBO objective the BW 
geometry gives the closest match to the true $x_2$ variance. The BW-VR-IWAE method follows the same trend, with the smaller power $\alpha = 0.1$ ($\widehat{\text{Var}}(x_2) = 3.6$) clearly outperforming $\alpha = 0.9$ ($1.1$). Overall, BW-IW-ELBO best recovers the target's dispersion, confirming that both the IW-ELBO objective and the BW geometry contribute to mass coverage. As expected, no method recovers the banana exactly, since its curved geometry lies outside the Gaussian variational family. 

We emphasize that this section should be viewed as a proof of concept rather than a comprehensive evaluation of BW-VR-IWAE. Its purpose is to demonstrate that the proposed Bures--Wasserstein geometry is not tied specifically to the IW-ELBO objective, but instead provides a more general geometric framework that can be incorporated into other importance-weighted variational objectives. A thorough study of BW-VR-IWAE would require addressing a distinct set of theoretical and empirical questions, including extensive evaluations on challenging models such as those considered by \citet{daudel2026importance}. Such an investigation would considerably broaden the scope of the current paper and is therefore left for future work.

\begin{figure}[h!]
    \centering
    \includegraphics[width=\textwidth]{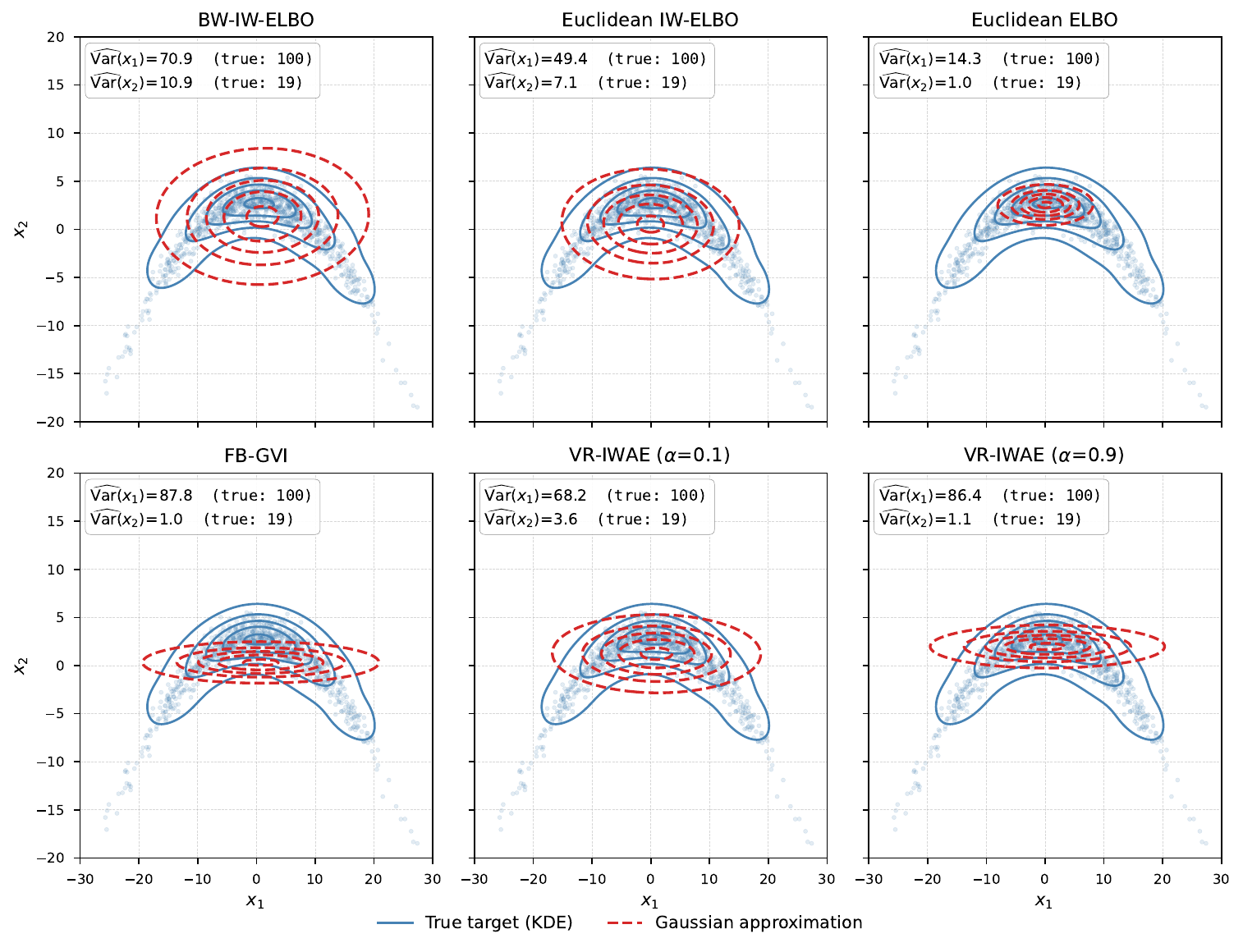}
    \caption{A comparison of different VI algorithms for banana-shaped distribution}
    \label{fig:banana}
\end{figure}

\section{Conclusion} \label{sec:conclusion}

This work bridges the gap between two distinct paradigms, IW-ELBO in importance weighted VI and WGF in optimal transport, by casting the optimization of IW-ELBO as a problem defined on the Wasserstein and BW space.
This formulation reveals the advantageous SNR scaling properties of the Wasserstein and BW gradients for the IW-ELBO.
Consequently, this elucidates the theoretical mechanisms that render optimization in this geometric framework more stable and efficient than conventional Euclidean methods.
Building upon this foundation, we introduce the BW-IW-ELBO method, enabling tractable IW-ELBO optimization for the Gaussian family under the BW geometry.
We evaluated the resulting Gaussian approximations via extensive numerical experiments.
Our empirical assessment demonstrated that the BW-IW-ELBO VI exhibits both mass-covering behavior and high approximation accuracy, particularly in complex multimodal settings where standard Gaussian variational methods typically underperform.

From a practical perspective, using a VI method requires 
the choice of three main ingredients: the variational family, the variational objective, and the optimization method.
The interplay between these components can substantially affect the final approximation accuracy, posterior uncertainty quantification, and predictive performance.
Focus on the setting where the variational family is restricted to Gaussian approximations, we find that coupling the IW-ELBO objective with the BW gradient optimizer consistently improves the quality of the variational approximation, yielding better posterior coverage and higher-quality importance weights across the experiments considered.
Therefore, for practitioners using Gaussian variational families, our results suggest that the IW-ELBO objective combined with the BW optimization geometry provides an effective and practically robust approach to VI.

Several promising avenues for future research remain.
Establishing a rigorous theoretical analysis of the convergence properties of the BW-IW-ELBO VI would further bolster the methodological foundation of our approach.
Furthermore, extending IW-ELBO optimization to broader geometric spaces beyond the BW manifold may yield even more robust variational frameworks and a deeper understanding of the interplay between optimal transport geometry and importance-weighted VI.
From the practical viewpoint, evaluating Wasserstein-based VI algorithms on a comprehensive and diverse suite of problems \citep[e.g.,][]{Magnusson2025} is a valuable avenue of dedicated future work.

\bibliographystyle{apalike}
\bibliography{bibliography}	

\clearpage
\appendix

\begin{center}
	\LARGE \textbf{Appendix}
\end{center}

\vspace{40pt}

This appendix contains additional details and deferred proofs of the theoretical results presented in the main text.
\Cref{apx:proofs} provides the proof of the theoretical results.
\Cref{sec:BW_derivation_VE_IWAE} shows the detail derivation of the BW gradient of the VR-IWAE bound.
\Cref{sec:Additional Experimental Details} contains the additional experimental details.

\section{Proofs} \label{apx:proofs}

\subsection{Proof of Proposition \ref{pro: proposition1}}

\begin{proof}[Proof of Proposition \ref{pro: proposition1}]
	By exchanging the order of the expectation and decomposing the summation, the objective $\mathcal{F}_n(q_n)$ can be expressed as
	\begin{align}
		\F_n(q_n) = \mathop{\E}_{ \{ z_1, \dots, z_K \} \setminus z_n \overset{i.i.d.}{\sim} q } \Bigg[ \underbrace{ \mathop{\E}_{z_n \sim q_n} \left[ \log \left( \frac{1}{K} \frac{p(x, z_n)}{q_n(z_n)} + \frac{1}{K} \sum_{i \ne n} w(z_i) \right) \right] }_{ =: \mathcal{F}_n^*(q_n) } \Bigg],
	\end{align}
    where the dependency of $\mathcal{F}_n^*$ on the set of the variables $\{ z_1, \dots, z_K \} \setminus z_n$ is made implicit.
    Provided that the derivative and the expectation are interchangeable, the Wasserstein gradient of $\F_n$ can be derived as the Wasserstein gradient of $\mathcal{F}_n^*$ averaged over the set of the variables $\{ z_1, \dots, z_K \} \setminus z_n$. 
    To this end, we derive the Wasserstein gradient of $\mathcal{F}_n^*$.
    Recall the derivation of the Wasserstein gradient recapped in \Cref{sec:bures_wasserstein_geometry}.
    The first variation of $\mathcal{F}_n^*$ at $q_n = q$ can be written as
	\begin{align}
		\frac{d}{d \epsilon} \mathcal{F}_n^*(q + \epsilon \nu) \Big|_{\epsilon = 0} = \frac{d}{d \epsilon} \int_{\Z} \underbrace{ \log \left( \frac{1}{K} \frac{ p(x, z_n) }{ q(z_n) + \epsilon \nu(z_n) } + \frac{1}{K} \sum_{i \ne n} w(z_i) \right) }_{ = (*_1)} \underbrace{\vphantom{\log \left( \frac{1}{K} \frac{ p(x, z_n) }{ q(z_n) } + \frac{1}{K} \sum_{i \ne n} w(z_i) \right)} ( q(z_n) + \epsilon \nu(z_n) ) }_{ = (*_2)} d z_n \Bigg|_{\epsilon = 0}, 
	\end{align}
	where the right derivative in \eqref{eq:first_variation} equals the standard derivative at $\epsilon = 0$ under sufficient regularity. 
    By exchanging the order of the integral and derivative, it further follows that
	\begin{align}
		\frac{d}{d \epsilon} \mathcal{F}_n^*(q_n + \epsilon \nu) \Big|_{\epsilon = 0} & =  \int_{\Z} \left( \left( \frac{d}{d \epsilon} (*_1) \right) \times (*_2) + (*_1) \times \left( \frac{d}{d \epsilon} (*_2) \right) \right) \Bigg|_{\epsilon = 0} d z_n .
	\end{align}
    By standard calculation of derivatives, we have
	\begin{align}
    	(*_1) \Big|_{\epsilon = 0} & = \log \left( \frac{1}{K} \sum_{i=1}^{K} w(z_i) \right) , \\
		(*_2) \Big|_{\epsilon = 0} & = q(z_n) , \\
		\frac{d}{d \epsilon} (*_1) \Big|_{\epsilon = 0} & = -\frac{ w(z_n) }{ \sum_{i=1}^{K} w(z_i) } \frac{ \nu(z_n) }{ q(z_n) } , \\
		\frac{d}{d \epsilon} (*_2) \Big|_{\epsilon = 0} & = \nu(z_n).
	\end{align}
	Therefore, the first variation of the functional $\mathcal{F}_n^*(q_n)$ at $q_n = q$ results in the following form
	\begin{align}
		\delta \mathcal{F}_n^*(z_n) = \log \left( \frac{1}{K} \sum_{i=1}^{K} w(z_i) \right) - \frac{ w(z_n) }{ \sum_{i=1}^{K} w(z_i) } =: f(z_n), 
	\end{align}
    where the dependency of the function $f$ on the set of the variables $\{ z_1, \dots, z_K \} \setminus z_n$ remains made implicit.
    As the Wasserstein gradient of $\mathcal{F}_n^*$ is the gradient of the function $f$, we have that
	\begin{align}
		\nabla_{z_n} f(z_n) & = \frac{\nabla_{z_n} w(z_n) }{ \sum_{i=1}^{K} w(z_i) }-\frac{\nabla_{z_n} w(z_n) ( \sum_{i=1}^{K} w(z_i) ) - w(z_n) ( \nabla_{z_n} \sum_{i=1}^{K} w(z_i) )}{ ( \sum_{i=1}^{K} w(z_i) )^2 }  \\
		& = \frac{\nabla_{z_n} w(z_n) }{ \sum_{i=1}^{K} w(z_i) } - \frac{\nabla_{z_n} w(z_n) }{ \sum_{i=1}^{K} w(z_i) } + \frac{w(z_n) }{ \sum_{i=1}^{K} w(z_i) }\frac{\nabla_{z_n} w(z_n) }{ \sum_{i=1}^{K} w(z_i) } \\
		& = \left( \frac{ w(z_n) }{ \sum_{i=1}^{K} w(z_i)} \right)^2 \nabla_{z_n} \log w(z_n) .
	\end{align}
    Hence, taking the expectation of this Wasserstein gradient $\nabla_{z_n} f(z_n)$ over the implicit dependent variables $\{ z_1, \dots, z_K \} \setminus z_n$ yields the intended form of the Wasserstein gradient of $\F_n(q_n)$.
\end{proof}

\subsection{Proof of \Cref{pro: proposition2}}

\begin{proof}[Proof of \Cref{pro: proposition2}]
    The result is a direct implication of \eqref{eq:BW gradient derivation} and \eqref{eq:Wasserstein grad IW-ELBO}.
\end{proof}

\subsection{Proof of \Cref{the: Wass grad SNR}}
\begin{proof}[Proof of \Cref{the: Wass grad SNR}]
    First, as the $M$ sets $\{z_{m,2:K}\}_{m=1}^M$ are i.i.d copies of $z_{2:K}$, it is clear that $\text{SNR}\big(g^{z_r}_{M,K}(z)\big)=\sqrt{M}\times\text{SNR}\big(g^{z_r}_{1,K}(z)\big)$.
    Also, we will evaluate $\text{SNR}\big(g^{z_r}_{1,K+1}(z)\big)$ rather than $\text{SNR}\big(g^{z_r}_{1,K}(z)\big)$, as both have the same magnitude as $K\to\infty$. Write the former derivative as    
\begin{align}
    g^{z_r}_{1,K+1}(z)& =\left( \frac{ w(z) }{w(z)+\sum_{i=1}^{K} w_i(z_{i})} \right)^2 \frac{\partial}{\partial z_r}\log w(z)\\
    & =\left( \frac{ w(z)/K}{w(z)/K+p(x)+\frac1K\sum_{i=1}^{K}\big(w_i(z_{i})-p(x)\big)} \right)^2 \frac{\partial}{\partial z_r}\log w(z)\\
    & =\left( \frac{ w(z)/K}{c_K+\frac1K\sum_{i=1}^{K}\xi_i} \right)^2 \frac{\partial}{\partial z_r}\log w(z),
\end{align}    
where $c_K:=w(z)/K+p(x)$ and $\xi_i:=w_i(z_{i})-p(x)$. We note that $\xi_i$ are i.i.d. with $\E(\xi_i)=\int p(x,z)dz-p(x)=0$.
Denote $X_K=1/\big(c_K+\frac1K\sum_{i=1}^{K}\xi_i\big)^2=1/\big(c_K+\bar\xi_K\big)^2$ with $\bar\xi_K=\frac1K\sum_{i=1}^{K}\xi_i$, we have
\[\text{SNR}\big(g^{z_r}_{1,K+1}(z)\big)=\text{SNR}(X_K).\]
Simple calculation gives
\begin{equation}\label{eq:helper eq 1}
\E(\bar{\xi}_K^2)=\frac{1}{K}\E(\xi_1^2),\;\;\;\E(\bar\xi_K^4)=\frac{1}{K^3}\E(\xi_1^4)+\frac{6(K-1)}{K^3}\big(\E(\xi_1^2)\big)^2.
\end{equation}
Applying the Mean Value Theorem to function $f(x)=1/(c_K+x)^2$,
\[f(x)=f(0)+f'(\zeta)x=c_K^{-2}-2(c_K+\zeta)^{-3}x,\]
where $\zeta$ is a value between $0$ and $x$, we have
\[X_K=c_K^{-2}-\eta_K\bar\xi_K\]
with $\eta_K=2(c_K+\zeta_K)^{-3}$ for some $\zeta_K$ between 0 and $\bar\xi_K$.
We have $c_K^{-2}=p(x)^{-2}+O(1/K)$, and will show that $\E(\eta_K\bar\xi_K)=O(1/\sqrt{K})$.
As $\zeta_K$ is between 0 and $\bar\xi_K$, $p(x)+\zeta_K$ is between $p(x)$ and $\frac{1}{K}\sum_iw_i(z_i)$, 
\begin{align}
    \eta_K^2=\frac{4}{\Big(w(z)/K+p(x)+\zeta_K\Big)^6}&\leq\frac{4}{\Big(p(x)+\zeta_K\Big)^6}\\
    &\leq 4\max\Big\{\frac{1}{p(x)^6},\frac{1}{\Big(\frac{1}{K}\sum_i w_i(z_i)\Big)^6}\Big\}\\
    &\leq 4\max\Big\{\frac{1}{p(x)^6},\frac{1}{K}\sum_i\frac{1}{w_i(z_i)^6}\Big\}, 
\end{align}
the last inequality is because $x\mapsto x^{-6}$ is convex in $(0,\infty)$.
This gives
\begin{equation}
    \E(\eta_K^2)\leq 4\max\Big\{\frac{1}{p(x)^6},\E(w(Z)^{-6})\Big\}:=C<\infty
\end{equation}
because of Assumption (ii). By the Cauchy-Schwarz theorem, and then using Assumption (i) to bound $\E(\xi_1^2)$,
\begin{equation}\label{eq: term eta K xi K}
    |\E(\eta_K\bar\xi_K)|\leq \Big(\E(\eta_K^2)\E(\bar\xi_K^2)\Big)^{1/2}\leq \Big(\frac{C}{K}\E(\xi_1^2)\Big)^{1/2}=O(\frac{1}{\sqrt{K}}).
\end{equation}
Now,
\begin{equation}
    |\E(X_K)| = \big|c_K^{-2}-\E(\eta_K\bar\xi_K)\big|\geq \big|c_K^{-2}-|\E(\eta_K\bar\xi_K)|\big|.
\end{equation}
As $c_K^{-2} >0$ and $|\E(\eta_K\bar\xi_K)|\to0$, $c_K^{-2}-|\E(\eta_K\bar\xi_K)|>0$ as $K$ is large enough, then
\begin{equation}
    |\E(X_K)| \geq c_K^{-2}-|\E(\eta_K\bar\xi_K)| =  1/p(x)^2 + o(1),  
\end{equation}
which implies 
\begin{equation}\label{eq: order of mean X}
    |\E(X_K)| = \Omega(1).
\end{equation}
For the variance term
\begin{align}
    \V(X_K) =\E(X_K^2)-\E(X_K)^2 =\E(\eta_K^2\bar\xi_K^2)-\big(\E(\eta_K\bar\xi_K)\big)^2\leq \E(\eta_K^2\bar\xi_K^2).
\end{align}
Similar to the above, we can show that
\begin{equation}
    \E(\eta_K^4)\leq 16 \max\Big\{\frac{1}{p(x)^{12}},\E(w(Z)^{-12})\Big\}<\infty.
\end{equation}
Hence, by \eqref{eq:helper eq 1},
\begin{equation}  
\V(X_K)\leq \E(\eta_K^2\bar\xi_K^2)\leq\Big(\E(\eta_K^4)\E(\bar\xi_K^4)\Big)^{1/2}=O(\frac{1}{K}),
\end{equation}
which implies 
\begin{equation}\label{eq: order of var X}
    \frac{1}{\V(X_K)} = \Omega({K}).
\end{equation}
Combining \eqref{eq: order of mean X}-\eqref{eq: order of var X} together, we obtain
\[\text{SNR}(X_K)=\frac{|\E(X_K)|}{\sqrt{\V(X_K)}}=\Omega(\sqrt{K}).\]
This completes the proof.
\end{proof}

\subsection{Proof of \Cref{pro: BW grad SNR}}
\begin{proof}[Proof of \Cref{pro: BW grad SNR}]
Let $\widehat{a}_*^{(M,K)}$ and $\widehat{S}_*^{(M,K)}$ denote the Monte Carlo estimator of $a_*$ and $S_*$, where the former is $d$-dimensional vector and the latter is $d \times d$ matrix. 
Thus, $g_{M,K}$ denotes either the $r$-th entry of $\widehat{a}_*^{(M,K)}$ or the $(i,j)$-th entry of $\widehat{S}_*^{(M,K)}$ for arbitrary entry index.
Since the estimator $g_{M,K}$ is an average based on $M$ i.i.d. samples from the $K$ variables $z_{1:K}$, it is clear that $\text{SNR}\big(g_{M,K}\big)=\sqrt{M}\times\text{SNR}\big(g_{1,K}\big)$.
Hence, it suffices to consider $M=1$. 
For notational convenience, we evaluate $\text{SNR}\big(g_{1,K+1}\big)$ rather than $\text{SNR}\big(g_{1,K}\big)$ with no loss of generality.

Let $z_{1:K}\sim \mathcal{N}(m,\Sigma)$ denote the first $K$ samples.
Let $z\sim \mathcal{N}(m,\Sigma)$ denote the last $K+1$-th sample.
The BW gradient $\mathcal{F}_{K+1}(q)$, based on the $K+1$ samples, at a density $q = \mathcal{N}(m, \Sigma)$ is given by
\begin{align}
    \nabla^{\text{BW}}[\mathcal{F}_{K+1}(q) ](u) = a_* + S_* (u - m) ,
\end{align}	
with
\begin{align}
    a_* & = \mathop{\E}_{z,z_1, \dots, z_K \overset{i.i.d.}{\sim} \mathcal{N}(m,\Sigma)} \left[ \left( \frac{ w(z) }{ w(z)+\sum_{i=1}^{K} w(z_i)} \right)^2 \nabla_{z} \log w(z) \right]\\
    S_* & = \mathop{\E}_{z,z_1, \dots, z_K \overset{i.i.d.}{\sim} \mathcal{N}(m,\Sigma)} \left[ \nabla_{z}\left\{ \left(\frac{ w(z) }{ w(z)+\sum_{i=1}^{K} w(z_i)} \right)^2 \nabla_{z} \log w(z)\right\} \right] .
\end{align}	
First, we consider the case where $g_{1,K+1}$ is a component of $\widehat{a}_*^{(1,K+1)}$.
This can be written as
\begin{align}
    g_{1,K+1}& =\left( \frac{ w(z) }{w(z)+\sum_{i=1}^{K} w(z_{i})} \right)^2 \partial_{z_r}\log w(z)
\end{align}    
where $z,z_{1:K}\sim \mathcal{N}(m,\Sigma)$, and $\partial_{z_r}\log w(z)$ is the partial derivative of $\log w(z)$ w.r.t some component $z_r$ of vector $z$.
We express this term as
\begin{align}
    g_{1,K+1}& =\frac{1}{K^2}\underbrace{w(z)^2\partial_{z_r}\log w(z)}_{f(z)}\underbrace{\left(\frac{1}{c_K(z)+\frac1K\sum_{i=1}^{K}\xi_i}\right)^2}_{X_K}=\frac{1}{K^2}f(z)X_K,
\end{align}    
where, as in the proof of \Cref{the: Wass grad SNR}, we set $c_K(z)=w(z)/K+p(x)$, $\xi_i=w(z_{i})-p(x)$, and 
\[X_K=c_K(z)^{-2}-\eta_K\bar\xi_K\]
with $\eta_K=2(c_K(z)+\zeta_K)^{-3}$ for some $\zeta_K$ between 0 and $\bar\xi_K=\frac{1}{K}\sum\xi_i$.

We have
\[\E(g_{1,K+1})=\frac{1}{K^2}\E\Big(f(z)\E\big(X_K|z\big)\Big).\]
From the proof of \Cref{the: Wass grad SNR},
\begin{align}
    \eta_K^2
    &\leq 4\max\Big\{\frac{1}{p(x)^6},\frac{1}{K}\sum_i\frac{1}{w(z_i)^6}\Big\},\;\;\forall z 
\end{align}
hence $\E(\eta_K^2|z)\leq C$, for some constant $C$, for all $z$. Note that $\bar\xi_K$ is independent of $z$,
\[|\E(\eta_K\bar{\xi}_K|z)|\leq \Big(\E(\eta_K^2|z)\E(\bar\xi_K^2|z)\Big)^{1/2}\leq \Big(C\E(\bar\xi_K^2)\Big)^{1/2}=O(\frac{1}{\sqrt{K}}).\]
Then,
\begin{equation}\label{eq:evaluate E(X_K|z)}
    \E(X_K|z)=c_K(z)^{-2}+O(\frac{1}{\sqrt{K}})=\frac{1}{p(x)^2}+O(\frac{w(z)}{K})+O(\frac{1}{\sqrt{K}}).
\end{equation}
Hence, by the assumption that $\E(f(z))<\infty$ and $\E(w(z)f(z))<\infty$, we have
\begin{equation}\label{eq:mean term of g_{1,K+1}}
    \E(g_{1,K+1})=\frac{1}{K^2}\frac{\E(f(z))}{p(x)^2}+o(\frac{1}{K^2}).
\end{equation}
For the variance term,
\begin{equation}
    \V(g_{1,K+1})=\frac{1}{K^4}\Big[\E\big(f(z)^2\V(X_K|z)\big)+\V\big(f(z)\E(X_K|z)\big)\Big].
\end{equation}
Again, from the proof of \Cref{the: Wass grad SNR}, $\V(X_K|z)=O(1/K)$, hence
\begin{equation}\label{eq:variance term of g_{1,K+1}}
    \V(g_{1,K+1})=\frac{1}{K^4} \frac{1}{p(x)^4}\V(f(z)) + o( \frac{1}{K^4} ).
\end{equation}
The SNR is the ratio of \eqref{eq:mean term of g_{1,K+1}} and the square root of \eqref{eq:variance term of g_{1,K+1}}, yielding
\begin{align}
    \text{SNR}(g_{1,K+1}) = \frac{\Big|\frac{1}{K^2}\frac{1}{p(x)^2}\E(f(z))+o(\frac{1}{K^2})\Big|}{\sqrt{\frac{1}{K^4} \frac{1}{p(x)^4}\V(f(z))+o( \frac{1}{K^4} ) )}} \ge \frac{ \Omega\Big( \frac{1}{K^2} \Big| \E(f(z)) \Big| \Big) }{ \frac{1}{K^2} \sqrt{ \V(f(z)) } } = \Omega(1) , 
\end{align}
where the last equality holds because $\E(f(z))$ is non-zero by the assumption (ii) and $\V(f(z))$ is positive.

Next, we consider the case where $g_{1,K+1}$ is a component of the matrix $\widehat{S}_*^{(1,K+1)}$, i.e., the $(i,j)$-th entry of the matrix for arbitrary $i,j$.
For notational convenience, define
\[\overline{W}_K := \frac{w(z)}{w(z)+\sum_{i=1}^Kw(z_i)} . \]
Standard algebraic operation yields
\begin{align}
    S_* = \mathop{\E}_{z,z_{1:K}} \left[\overline{W}_K^2\Big(\nabla^2_z\log w(z)+2\nabla_z\log w(z)(\nabla_z\log w(z))^\top\Big)-2\overline{W}_K^3\nabla_z\log w(z)(\nabla_z\log w(z))^\top\right] .
\end{align}
For $i$ and $j$ components $z_i,z_j$ of $z$, the component $g_{1,K+1}$ can be written as
\begin{align}
    g_{1,K+1} = \underbrace{\overline{W}_K^2\Big(\partial^2_{z_iz_j}\log w(z)+2\partial_{z_i}\log w(z)\partial_{z_j}\log w(z)\Big)}_{(*)}-2\underbrace{\overline{W}_K^3\partial_{z_i}\log w(z)\partial_{z_j}\log w(z)}_{(**)} ,
\end{align}	
where $z,z_{1:K}\sim \mathcal{N}(m,\Sigma)$.
Observe that $\overline{W}_K^2=\big(w(z)/K\big)^2X_K$, which yields
\[(*)=\frac{1}{K^2}\underbrace{w(z)^2\big(\partial^2_{z_iz_j}\log w(z)+2\partial_{z_i}\log w(z)\partial_{z_j}\log w(z)\big)}_{f_1(z)}X_K=\frac{1}{K^2}f_1(z)X_K.\]
Thus, we have the following form of the expectation of $(*)$:
\begin{equation}\label{eq:term (*) mean}
    \E(*) = \frac{1}{K^2}\E\big(f_1(z)\E(X_K|z)\big) = \frac{1}{K^2}\frac{\E(f_1(z))}{p(x)^2}+o(\frac{1}{K^2})    
\end{equation}
where we have used \eqref{eq:evaluate E(X_K|z)} and the assumptions that $\E(f_1(z))<\infty$ and $\E(w(z)f_1(z))<\infty$.
We move on to the evaluation of term $(**)$.
We express the cubed term $\overline{W}_K^3$ as follows:
\[\overline{W}_K^3=\Big(\frac{w(z)}{K}\Big)^3\frac{1}{\big(c_K(z)+\bar\xi_K\big)^3}=\Big(\frac{w(z)}{K}\Big)^3Y_K, \quad\text{with}\quad Y_K:=\frac{1}{\big(c_K(z)+\bar\xi_K\big)^3}.\]
Using Taylor's expansion for the function $x:\mapsto (c_K(z)+x)^{-3}$, we obtain
\[Y_K=c_K(z)^{-3}-\kappa_K\bar\xi_K\]
with $\kappa_K=3(c_K(z)+\zeta_K)^{-4}$ for some $\zeta_K$ between 0 and $\bar\xi_K$.
Similar to the proof of \Cref{the: Wass grad SNR}, $\E(\kappa_K^2|z)$ and $\E(\kappa_K^4|z)$ are bounded independent of $z$, hence
\begin{equation}\label{eq:bounded result 1}
    |\E(\kappa_K\bar\xi_K|z)|\leq \Big(\E(\kappa_K^2|z)\E(\bar\xi_K^2)\Big)^{1/2}=O(\frac{1}{K}), \quad |\E(\kappa_K^2\bar\xi_K^2|z)|\leq \Big(\E(\kappa_K^4)\E(\bar\xi_K^4|z)\Big)^{1/2}=O(\frac{1}{K}) .
\end{equation}    
Leveraging this rate, it is straightforward to verify 
\begin{equation}
    \E(Y_K|z)=\frac{1}{p(x)^3}+o(1),\quad\text{and}\quad \V(Y_K|z)=O(\frac{1}{K}).
\end{equation}
We further express term $(**)$ as
\[(**) = \frac{1}{K^3}\underbrace{w(z)^3\partial_{z_i}\log w(z)\partial_{z_j}\log w(z)}_{f_2(z)}Y_K=\frac{1}{K^3}f_2(z)Y_K , \]
which in turn yields
\begin{equation}\label{eq:term (**) mean}
    \E(**) = \frac{1}{K^3}\E\big(f_2(z)\E(Y_K|z)\big) = \frac{1}{K^3}\frac{\E(f_2(z))}{p(x)^3}+o\big(\frac{1}{K^3}\big) .
\end{equation}
Combining \eqref{eq:term (*) mean}-\eqref{eq:term (**) mean} gives,
\begin{equation}\label{eq:mean of g_{1,K+1}} 
    \E(g_{1,K+1}) = \frac{1}{K^2}\frac{\E(f_1(z))}{p(x)^2} + o(\frac{1}{K^2}) .
\end{equation}
For the variance of term $(*)$, we have
\begin{align}
\V(*) & =\frac{1}{K^4}\Big(\E\big(f_1(z)^2\V(X_K|z)\big)+\V(f_1(z)\E(X_K|z))\Big) \\
& =\frac{1}{K^4}\Big(O(\frac{1}{K})+\frac{1}{p(x)^4}\V(f_1(z))+o(1)\Big) \\
& = \frac{1}{K^4} \frac{1}{p(x)^4} \V(f_1(z)) + o(\frac{1}{K^4}).
\end{align}
A similar derivation provides
\[\V(**)=O(\frac{1}{K^6}).\]
The Cauchy-Schwarz inequality implies 
\begin{equation}\label{eq:var of g_{1,K+1}}
    \V(g_{1,K+1})= \frac{1}{K^4} \frac{1}{p(x)^4} \V(f_1(z)) + o(\frac{1}{K^4}) .
\end{equation}
From \eqref{eq:mean of g_{1,K+1}} and \eqref{eq:var of g_{1,K+1}}, we conclude that 
\[\text{SNR}(g_{1,K+1}) = \frac{|\E(g_{1,K+1})|}{\sqrt{\V(g_{1,K+1})}} = \frac{| \frac{1}{K^2}\frac{\E(f_1(z))}{p(x)^2} + o(\frac{1}{K^2}) |}{ \sqrt{\frac{1}{K^4} \frac{1}{p(x)^4} \V(f_1(z)) + o(\frac{1}{K^4})} } \ge \frac{ \Omega\Big( \frac{1}{K^2} \Big| \E(f_1(z)) \Big| \Big) }{ \frac{1}{K^2} \sqrt{ \V(f_1(z)) } } = \Omega(1) , \]
where the last equality holds because $\E(f_1(z))$ is non-zero by the assumption (iii) and $\V(f_1(z))$ is positive.
This completes the proof.
\end{proof}

\subsection{Proof of \Cref{pro: proposition VR-IWAE}}

\begin{proof}[Proof of \Cref{pro: proposition VR-IWAE}]
    By exchanging the order of the expectation, $\G_n(q_n)$ can be expressed as
    \begin{align}
		\G_n(q_n) & = \frac{1}{1-\alpha} \mathop{\E}_{z_n \sim q_n} \Bigg[ \underbrace{ \mathop{\E}_{\{z_1, \dots, z_K\} \setminus z_n \overset{i.i.d.}{\sim} q} \bigg[ \log \bigg( \frac{1}{K} \frac{p(x, z_n)^{1-\alpha}}{q_n(z_n)^{1-\alpha}} + \frac{1}{K} \sum_{i \ne n} w(z_i)^{1-\alpha} \bigg) \bigg] }_{=: \G_n^*(q_n) } \Bigg] 
	\end{align}
    where the dependency of $\mathcal{G}_n^*$ on the set of the variables $\{ z_1, \dots, z_K \} \setminus z_n$ is made implicit.
    Provided that the derivative and the expectation are interchangeable, the Wasserstein gradient of $\G_n$ can be derived as the Wasserstein gradient of $\mathcal{G}_n^*$ averaged over the set of the variables $\{ z_1, \dots, z_K \} \setminus z_n$. 
    To this end, we derive the Wasserstein gradient of $\mathcal{G}_n^*$.
    Recall the derivation of the Wasserstein gradient recapped in \Cref{sec:bures_wasserstein_geometry}.
    The first variation of $\mathcal{G}_n^*$ at $q_n = q$ can be written as
	\begin{align}
		\frac{d}{d \epsilon} \G_n^*(q + \epsilon \nu) \Big|_{\epsilon = 0} = \frac{d}{d \epsilon} \int_{\Z} \underbrace{ \log \Bigg( \frac{1}{K} \frac{ p(x, z_n)^{1-\alpha} }{ ( q(z_n) + \epsilon \nu(z_n) )^{1-\alpha} } + \frac{1}{K} \sum_{i \ne n} w(z_i)^{1-\alpha} \Bigg) }_{ = (*_1)} \underbrace{\vphantom{\log \Bigg( \frac{1}{K} \sum_{i \ne n}} ( q(z_n) + \epsilon \nu(z_n) ) }_{ = (*_2)} d z_n \Bigg|_{\epsilon = 0}  
	\end{align}
    where the right derivative in \eqref{eq:first_variation} equals the standard derivative at $\epsilon = 0$ under sufficient regularity. 
    By exchanging the order of the integral and derivative, it further follows that
	\begin{align}
		\frac{d}{d \epsilon} \G_n^*(q + \epsilon \nu) \Big|_{\epsilon = 0} & = \int_{\Z} \left( \left( \frac{d}{d \epsilon} (*_1) \right) \times (*_2) + (*_1) \times \left( \frac{d}{d \epsilon} (*_2) \right) \right) \Bigg|_{\epsilon = 0} \times d z_n .
	\end{align}
    By standard calculation of derivatives, we have
    \begin{align}
        (*_1) |_{\epsilon = 0} & = \log \left( \frac{1}{K} \sum_{i=1}^{K} w(z_i)^{1-\alpha}\right) , \\
		(*_2) |_{\epsilon = 0} & = q(z_n) , \\
        \frac{d}{d \epsilon} (*_1) \Big|_{\epsilon = 0} & = -(1-\alpha) \frac{ w(z_n)^{1-{\alpha}} }{ \sum_{i=1}^{K} w(z_i)^{1-\alpha} } \frac{ \nu(z_n) }{ q(z_n) } , \\
		\frac{d}{d \epsilon} (*_2) \Big|_{\epsilon = 0} & = \nu(z_n).
	\end{align}
    Therefore, the first variation of the functional $\G_n^*$ at $q_n = q$ results in the following form
	\begin{align}
		\delta \G_n^*(z_n) = \log \left( \frac{1}{K} \sum_{i=1}^{K} w_i(z_i)^{1-\alpha} \right)-(1-\alpha) \frac{ w_n(z_n)^{1-\alpha}}{ \sum_{i=1}^{K} w_i(z_i)^{1-\alpha} } =: g(z_n)
	\end{align}
    where the dependency of the function $g$ on the set of the variables $\{ z_1, \dots, z_K \} \setminus z_n$ remains made implicit.
    As the Wasserstein gradient of $\mathcal{G}_n^*$ is the gradient of the function $g$, we have that
	\begin{align}
		\nabla_{z_n} g(z_n) & = (1 - \alpha) \lr{ \alpha \frac{ w(z_n)^{1-\alpha} }{ \sum_{i=1}^{K} w(z_i)^{1-\alpha}} + (1-\alpha) \bigg ( \frac{  w(z_n)^{1-\alpha}}{ \sum_{i=1}^{k} w(z_i)^{1-\alpha} } \bigg)^2 } \nabla_{z_n} \log w(z_n) . 
	\end{align}
    Hence, taking the expectation of this Wasserstein gradient $\nabla_{z_n} g(z_n)$ over the implicit dependent variables $\{ z_1, \dots, z_K \} \setminus z_n$ yields the intended form of the Wasserstein gradient of $\G_n(q_n)$.
\end{proof}

\section{Derivations of BW gradient for VR-IWAE bound} \label{sec:BW_derivation_VE_IWAE}

This section presents the explicit expression of the BW gradient of the VR-IWAE bound, together with the detailed derivation.
Recall that it suffices to consider the Wasserstein gradient of the VR-IWAE bound at the index $n = K$ as it remains invariant to the choice of the index $n$.
For notational convenience, let $G$ denotes the Wasserstein gradient of the VR-IWAE bound at the common variational distribution $q$, and let $g(z_1, \dots, z_K) := w(z_K)^{1-\alpha} / \sum_{i=1}^{K} w(z_i)^{1-\alpha}$.

We derive the two terms $a_*$ and $S_*$ characterizing the BW gradient of the negative VR-IWAE bound in \eqref{eq:BW_terms_VR_IWAE}.
It follows from \eqref{eq:Wass grad VR-IWAE} and the derivation of the BW gradient recapped in \Cref{sec:bures_wasserstein_geometry} that
\begin{align}
    a_* & = - \mathop{\E}_{z_K \sim q}[ G(z_K) ] \\
    & = - \mathop{\E}_{z_1, \dots, z_K \overset{i.i.d.}{\sim} q} \left[ \left( \alpha g(z_1, \dots, z_K) + (1-\alpha) g(z_1, \dots, z_K)^2 \right) \nabla_{z_K} \log w(z_K) \right].
\end{align}
To establish $S_*$, we derive the gradient of $g(z_1, \dots, z_K)$ with respect to $z_K$:
\begin{align}
    \nabla_{z_K} g(z_1, \dots, z_K) & = \frac{\nabla_{z_K} w(z_K)^{1-\alpha}}{\sum_{i=1}^{K} w(z_i)^{1-\alpha}} - \frac{w(z_K)^{1-\alpha} \nabla_{z_K} w(z_K)^{1-\alpha}}{(\sum_{i=1}^{K} w(z_i)^{1-\alpha})^2} \\
    & = (1-\alpha) \left( \frac{w(z_K)^{1-\alpha}}{\sum_{i=1}^{K} w(z_i)^{1-\alpha}} - (\frac{w(z_K)^{1-\alpha}}{\sum_{i=1}^{K} w(z_i)^{1-\alpha}})^2 \right) \nabla_{z_K} \log w(z_K) \\
    & = (1-\alpha) ( g(z_1, \dots, z_K) - g(z_1, \dots, z_K)^2 ) \nabla_{z_K} \log w(z_K) .
\end{align}
This leads to the explicit expression of the remaining term $S_*$:
\begin{align}
    S_* & = - \mathop{\E}_{z_K \sim q_K}\left[ \nabla_{z_K} G(z_K) \right] \\
    & = - \mathop{\E}_{z_1, \dots, z_K \overset{i.i.d.}{\sim} q}\bigg[ \nabla_{z_K} \big( \underbrace{ ( \alpha g(z_1, \dots, z_K) + (1-\alpha) g(z_1, \dots, z_K)^2 ) }_{ =: (*_1)} \underbrace{ \nabla_{z_K} \log w(z_K) }_{ =: (*_2)} \big) \bigg] \\
    & = - \mathop{\E}_{z_1, \dots, z_K \overset{i.i.d.}{\sim} q}\left[ ( \nabla_K (*_1) ) (*_2) + (*_1) ( \nabla_K (*_2) ) \right] .
\end{align}
Here, using the expression of $\nabla_{z_K} g(z_1, \dots, z_K)$, we have
\begin{align}
    \nabla_K (*_1) & = \alpha \nabla_{z_K} g(z_1, \dots, z_K) + 2 (1 - \alpha) g(z_1, \dots, z_K) \nabla_{z_K} g(z_1, \dots, z_K) \\
    & = \Big( (1 - \alpha) ( g(z_1, \dots, z_K) - g(z_1, \dots, z_K)^2 ) (\alpha + 2 (1 - \alpha) g(z_1, \dots, z_K) ) \Big) \nabla_{z_K} \log w(z_K) .
\end{align}
Plugging $\nabla_K (*_1)$ in together with $\nabla_K (*_2) = \nabla_{z_K}^2 \log w(z_K)$ completes the derivation.

\section{Additional Experimental Details}\label{sec:Additional Experimental Details}

The covariance update
$\Sigma \gets (I + \eta \widehat{S}_*)\, \Sigma\, (I + \eta \widehat{S}_*)$
has the form $H \Sigma H$ with $H = I + \eta \widehat{S}_*$. Since $H$ is symmetric, if
$\Sigma$ is positive definite and $H$ is nonsingular, the updated covariance remains
positive definite. Numerical instability can occur when $H$ becomes nearly singular or has very
large eigenvalues, causing excessive shrinkage or expansion of $\Sigma$ along some
directions. We therefore apply eigenvalue clipping to $H$.

Before applying the update we
\begin{enumerate}
    \item symmetrise the Monte Carlo estimate, $\widehat{S}_* \gets \tfrac{1}{2}(\widehat{S}_* +
    \widehat{S}_*^{\top})$. The matrix $\widehat{S}_*$ is symmetric in exact
    arithmetic but may accumulate small asymmetric floating-point errors;
    \item form $H = I + \eta \widehat{S}_*$ and compute its eigendecomposition
    $H = V \Lambda V^{\top}$;
    \item clamp the eigenvalues of $H$ to $[\lambda_{\mathrm{low}},
    \lambda_{\mathrm{high}}]=[0.1,1.5]$ and reconstruct $H \gets V\, \Lambda_{\mathrm{clipped}}\,
    V^{\top}$.
\end{enumerate}
After clipping, $H$ is invertible with condition number at most
$\lambda_{\mathrm{high}}/\lambda_{\mathrm{low}}$, and the resulting iterate
$H \Sigma H^{\top}$ is positive definite in exact arithmetic.

The clipping thresholds have a direct interpretation as per-step multiplicative bounds on directional variances along eigenvectors of \(H\): along such a direction, a single update cannot divide the variance by more than \(100\) or multiply it by more than \(2.25\). Thus the clipping is used as a numerical safeguard against singular or explosive covariance updates, rather than as a primary tuning parameter.

All our experiments share the same basic optimization setting as shown in Table \ref{tab:hyper}.

\begin{table}[htbp]
\centering
\caption{Optimization settings for the experiment (Monte Carlo setting is only for \cref{sec:Simulation Study: Mass-covering Property} and \cref{sec:Application: Census Binary Data}).}
\label{tab:hyper}
\begin{tabular}{lcc}
\toprule
\textbf{Method} & \textbf{Monte Carlo setting} & \textbf{Additional safeguard} \\
\midrule
BW-IW-ELBO
& $M=100,\ K=5$
& $H$ eigenvalue clip $[0.1,1.5]$ \\
FB-GVI
& $M=500, K=1$
& - \\
Euclidean IW-ELBO
& $M=100,\ K=5$
& gradient clipping at norm $1.0$ \\
Euclidean ELBO
& $M=500, K=1$
& gradient clipping at norm $1.0$ \\
\bottomrule
\end{tabular}
\end{table}

\subsection{Empirical Details for \Cref{sec:Simulation Study: Mass-covering Property}}\label{sec:Experimental Details for Eggbox}

We use a four-component eggbox target with equally weighted Gaussian components ($w_k=1/4$). The component means and covariances are listed in Table \ref{tab:eggbox-target}. The analytical target moments are
\begin{equation}
m^\star = \frac{1}{4}\sum_{k=1}^{4}m_k,\qquad
\Sigma^\star =
\frac{1}{4}\sum_{k=1}^{4}\Sigma_k
+
\frac{1}{4}\sum_{k=1}^{4}
(m_k-m^\star)(m_k-m^\star)^\top.
\end{equation}

Step sizes were selected by a grid search, choosing the value with the highest mean final objective across seeds. Resulted step size $\eta$ for BW-IW-ELBO is 0.5, for FB-GVI is 0.1, and learning rate $\mathrm{lr}=0.05/0.01$ for Euclidean IW-ELBO and Euclidean ELBO, respectively. All methods were initialized at $m_0=[6.0,12.0]^\top$ and $\Sigma_0=5I_2$, run for $1000$ iterations, and used $500$ target evaluations per iteration. Results are averaged over $10$ random seeds.

The eigenvalue clipping for BW-IW-ELBO is rarely activated: it occurred around $0.0\%-0.1\%$ of the total iterations across the \(10\) evaluation seeds.

\begin{table}[t]
\centering
\small
\caption{Target eggbox: four equally weighted Gaussian components.}
\label{tab:eggbox-target}
\begin{tabular}{ccc}
\toprule
$k$ & $m_k$ & $\Sigma_k$ \\
\midrule
$1$ &
$[3,\,3]^\top$ &
$\begin{bmatrix}
1.0 & -0.8\\
-0.8 & 1.2
\end{bmatrix}$ \\[4pt]

$2$ &
$[5,\,5]^\top$ &
$\begin{bmatrix}
0.6 & 0.1\\
0.1 & 0.3
\end{bmatrix}$ \\[4pt]

$3$ &
$[7,\,2]^\top$ &
$\begin{bmatrix}
0.5 & -0.2\\
-0.2 & 0.3
\end{bmatrix}$ \\[4pt]

$4$ &
$[9,\,6]^\top$ &
$0.5\,I_2$ \\
\bottomrule
\end{tabular}
\end{table}

\subsection{Empirical Details for \Cref{sec:Simulation Study: Performance of gradient estimators}}\label{sec:Experimental Details for BLR experiment}

\textbf{SNR-scaling experiment}

We use a synthetic Bayesian logistic regression problem to evaluate the SNR
scaling of the three gradient estimators. For each latent dimension
$d\in\{20,50,80\}$, we generate a fixed dataset with $n=10$ observations.
Specifically, we draw $z_{\mathrm{true}}\sim\mathcal N(0,I_d)$ and
$\tilde X_i\sim\mathcal N(0,I_d)$, then set
$$
X_i = \sqrt{c/d}\,\tilde X_i,
\qquad
y_i \sim \mathrm{Bernoulli}(\sigma(X_i^\top z_{\mathrm{true}})),
$$
where $\sigma(t)=(1+\exp(-t))^{-1}$. We rescale the features with a fixed constant $c = 8$, chosen so that the prior-predictive logit variance $\mathrm{Var}(x_i^\top z_{\mathrm{true}}) = c$
is held constant across all $d$. This keeps the
logit scale comparable across dimensions. The prior is
$z\sim\mathcal N(0,I_d)$, and throughout this experiment the
variational distribution is fixed to
$
q(z)=\mathcal N(0,I_d)
$.
No optimization is performed in this SNR-scaling experiment; its goal is to isolate the Monte Carlo
behaviour of the gradient estimators.

Here, $X$ is the fixed design matrix and $y$ is the observations. Let
$
\ell(z) = \log p(y,z)-\log q(z) =\log w
$
denote the log importance weight. For the Wasserstein estimator, we evaluate
the gradient at the fixed point $z_0=0$. Given auxiliary samples
$z_1,\ldots,z_{K-1}\sim q$, the implemented estimator is
$$
\widehat g_{\mathrm W}(z_0)
=
\alpha_0^2 \nabla_z \ell(z_0),
\qquad
\alpha_0
=
\frac{\exp{\ell(z_0)}}
{\exp{\ell(z_0)}+\sum_{k=1}^{K-1}\exp{\ell(z_k)}} .
$$
For the BW-projected estimator, we draw $z_1,\ldots,z_K\sim q$ and compute
the affine components in~\eqref{eq:bw_iw_elbo1}--\eqref{eq:bw_iw_elbo2}.
The reported BW vector is formed by concatenating
$\widehat a_*$ and $\operatorname{diag}(\widehat S_*)$, so it has dimension $2d$. For the Euclidean estimator, the proposal is the parameter-dependent
Gaussian $q_m=\mathcal N(m,I_d)$ (covariance fixed to $I_d$), with the
corresponding log importance weight
$$
\ell_m(z)=\log p(y,z)-\log q_m(z).
$$
We use the reparameterized mean-gradient estimator (the IW-ELBO gradient of \cite{rainforth2018tighter}),
$$
\widehat g_{\mathrm E}
=
\nabla_m
\left[
\operatorname*{logsumexp}_{1\le k\le K}\,\ell_m(m+\epsilon_k)
-\log K
\right]\bigg|_{m=0},
\qquad
\epsilon_k\sim\mathcal N(0,I_d).
$$
For a given estimator, let $h_K$ denote one independent single-replicate
estimate. The $M$-replicate estimator is
$g_{M,K}=\frac{1}{M}\sum_{m=1}^M h_K^{(m)}$,
where the $h_K^{(m)}$ are independent. For the $K$-scaling experiment we fix
$M=1$ and use
$K\in\{10,100,200,500,1{,}000,2{,}000,4{,}000,8{,}000,10{,}000\}$.
For each pair $(K,d)$, each estimator, and each of $10$ random seeds, we draw
$200$ independent realisations, giving $2000$ pooled realisations for the
reported SNR. For the $M$-scaling experiment we fix $K=100$ and use $M\in\{1,2,4,8,16\}$.
For each $M$, estimator, dimension, and seed, we form $500$ independent
realisations of $g_{M,K}$, again pooling over $10$ seeds.

The SNR is computed coordinatewise. Given pooled realisations
$g^{(1)},\ldots,g^{(R)}\in\mathbb R^p$, where $p=d$ for the Wasserstein and
Euclidean estimators and $p=2d$ for the BW estimator, we compute
$$
\mathrm{SNR}_j
=
\frac{|\bar g_j|}{s_j},
\qquad
\bar g_j=\frac{1}{R}\sum_{r=1}^R g_j^{(r)},
$$
where $s_j$ is the sample standard deviation of the $j$th coordinate across
the pooled realisations. The scalar SNR reported in
Figure~\ref{fig:snr_scaling} and Table~\ref{tab:snr_scaling} is the average
of $\mathrm{SNR}_j$ over coordinates. Since the three estimators are defined
under different geometries, these SNR magnitudes are used only to estimate
scaling slopes, not to compare absolute estimator quality.

The $\log K$ slopes are obtained by ordinary least squares regression of
$\log\mathrm{SNR}$ on $\log K$ over the fitting window $K\ge 200$ at
$M=1$. The $\log M$ slopes are obtained by regressing $\log\mathrm{SNR}$ on
$\log M$ at fixed $K=100$. Standard errors are computed by a leave-one-seed-out
jackknife. If $\widehat z_{(-s)}$ is the fitted slope after omitting seed
$s$ and there are $S=10$ seeds, the reported standard error is
$$
\mathrm{SE}_{\mathrm{jack}}
=
\left[
\frac{S-1}{S}
\sum_{s=1}^S
\left(\widehat z_{(-s)}
-
\frac{1}{S}\sum_{r=1}^S \widehat z_{(-r)}
\right)^2
\right]^{1/2}.
$$

As a diagnostic for the finite-$K$ regime, we also compute the effective
sample size of the importance weights at $K=10{,}000$,
$$
\mathrm{ESS}
=
\frac{1}{\sum_{k=1}^K \bar w_k^2},
\qquad
\bar w_k
=
\frac{\exp{\ell(z_k)}}
{\sum_{j=1}^K \exp{\ell(z_j)}}.
$$
Averaging over $80$ independent draws gives ESS values
$229.6$, $210.0$, and $158.9$ for $d=20$, $50$, and $80$, respectively.
These values indicate that, even at $K=10{,}000$, the Euclidean estimator remains
in a pre-asymptotic importance-weight regime.

\textbf{Convergence }

The convergence study uses a low-dimensional synthetic Bayesian
logistic regression model, separate from the SNR-scaling experiment above.
We draw $n=1000$ observations with $d=3$ features (an intercept together with
two covariates $X_{i,j}\sim\mathcal N(0,1)$), generate labels
$y_i\sim\mathrm{Bernoulli}\!\left(\sigma(X_i^\top z^\star)\right)$ with
$z^\star=[1.04,\,0.61,\,-1.25]^\top$, and place a Gaussian prior
$z\sim\mathcal N(0,10\,I_3)$ on the coefficient $z$. The variational family is a full-covariance
Gaussian $q(z)=\mathcal N(m,\Sigma)$.

The eigenvalue clipping for BW-IW-ELBO is rarely activated: across the \(10\) evaluation seeds, it occurred in \(0.2\%\), \(1.7\%\), \(2.9\%\), and \(8.5\%\) of iterations at \(K=10,50,100,200\), respectively. The Euclidean baseline uses ADAM with tuning weights $(\beta_1, \beta_2) = (0.9, 0.999)$ and gradient-norm clipping at $1.0$; this is analogous to the BW eigenvalue clip in serving as a numerical safeguard, although it acts on the Euclidean parameter gradient rather than on the multiplicative covariance update.
 
For each method and each \(K\), step sizes are selected by a grid search over three held-out seeds \(\{9001,9002,9003\}\), disjoint from the evaluation seeds, with each configuration run for the full $3{,}000$ iterations. The grids are $\eta \in \{0.05, 0.1, 0.5, 1.0, 2.0, 4.0\}$ for BW and $\mathrm{lr} \in \{0.001, 0.005, 0.01, 0.02, 0.05, \\0.1, 0.5\}$ for the Euclidean baseline. A step is declared \emph{viable} if its mean final standard ELBO over the selection seeds lies within \(0.5\) nats of the best mean grid value and its worst selection-seed final ELBO is no more than \(50\) nats below that best mean value. Among viable steps we select the largest, which favors fast convergence at matched final quality. The selected values are $\eta = \{0.1, 1.0, 2.0, 4.0\}$ and $\mathrm{lr} = \{0.05, 0.02, 0.02, 0.01\}$ for $K = \{10, 50, 100, 200\}$. At \(K=200\), a post-selection probe at \(\eta=8.0\) diverged on two of the three selection seeds; in those runs, eigenvalue clipping was activated in \(47\%\)--\(48\%\) of iterations. The remaining seed converged with no clipping, placing the selected value \(\eta=4.0\) near the empirical stability boundary. Figure \ref{fig:heatmaps} reports the step-size selection landscape, measured by the mean final standard ELBO over the three selection seeds. The two methods show opposite trends as $K$ increases: the selected Euclidean learning rate decreases from $0.05$ to $0.01$, whereas the selected BW step increases from $0.1$ to $4.0$. This pattern is consistent with the BW gradient formula: increasing $K$ reduces the update magnitude through the squared normalized weights, so larger steps are needed, while the non-degenerate SNR keeps the estimator sufficiently stable.
 
Progress is measured by the standard ELBO of the current iterate \((m_t,\Sigma_t)\), estimated using \(n_{\mathrm{eval}}=500\) Monte Carlo samples with a fixed random seed; evaluation time is excluded from all timing measurements. The \emph{final ELBO} is the average of the last \(100\) raw evaluations. The \emph{common convergence threshold} is defined as $1$ nat below the best median final ELBO across both methods, giving \(-473.72\). A run is declared converged at the first iteration \(t\) such that the raw evaluation remains at or above this threshold for \(100\) consecutive iterations. All reported convergence statistics are medians and quartiles over $10$ evaluation seeds; every run of both methods reached the threshold.

Wall-clock time to threshold is the cumulative optimization time at the convergence iteration, excluding evaluation overhead. \Cref{tab:per_iter_cost} reports per-iteration costs. Both methods have forward costs that grow with \(M\times K\), but the Euclidean baseline backpropagates through the full \(M\times K\)-sample graph. The BW update avoids this backward pass, using the \(K\) samples for importance weights and evaluating one score and one closed-form Hessian per outer replicate.

\begin{table}[htbp]
    \centering
    \caption{Per-iteration cost (ms/iteration), mean $\pm$ standard
    deviation over the $10$ evaluation seeds, excluding ELBO-evaluation overhead.}
    \begin{tabular}{cccc}
        \toprule
        K & BW-IW-ELBO & Euclidean IW-ELBO & BW/Euc \\
        \midrule
        10  & $51.8 \pm 2.1$ & $70.4 \pm 3.5$  & 0.74x \\
        50  & $58.8 \pm 0.1$ & $85.5 \pm 0.3$  & 0.69x \\
        100 & $64.1 \pm 0.1$ & $94.2 \pm 0.7$  & 0.68x \\
        200 & $73.5 \pm 0.1$ & $112.4 \pm 0.2$ & 0.65x \\
        \bottomrule
    \end{tabular}
    \label{tab:per_iter_cost}
\end{table}

\begin{figure}[htbp]
    \centering
    \begin{subfigure}[b]{0.48\textwidth}
        \centering
        \includegraphics[width=\textwidth]{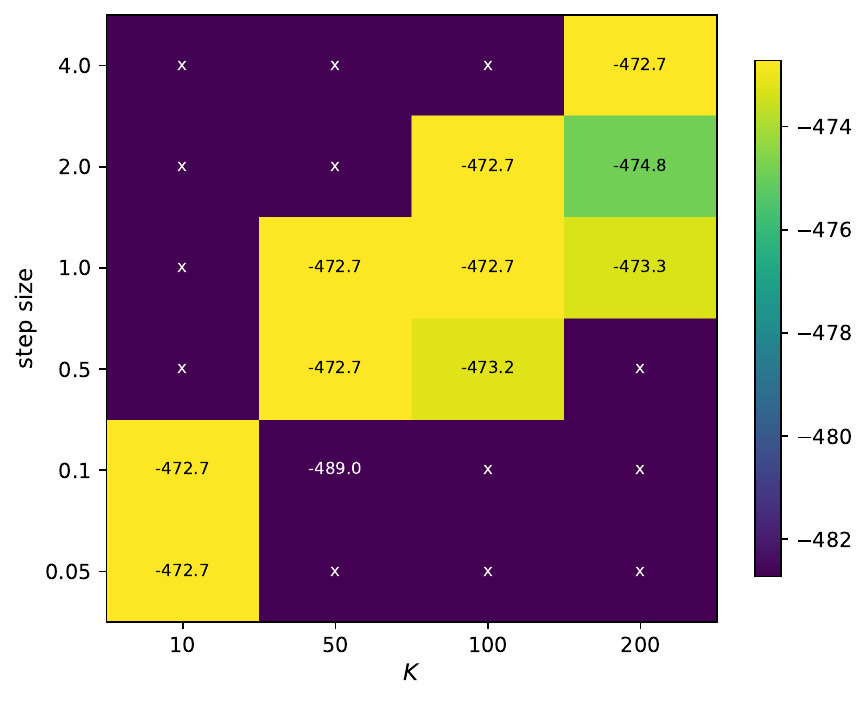}
        \caption{BW-IW-ELBO}
    \end{subfigure}
    \hfill 
    \begin{subfigure}[b]{0.48\textwidth}
        \centering
        \includegraphics[width=\textwidth]{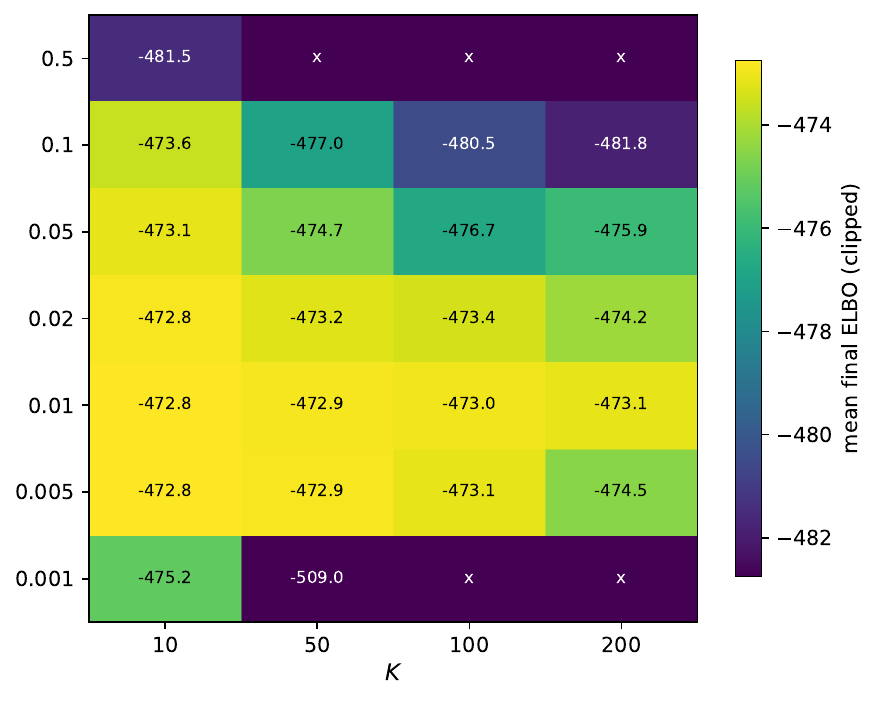}
        \caption{Euclidean IW-ELBO}
    \end{subfigure}    
    \caption{Step-size selection landscape: mean final standard ELBO over three selection seeds for every $(K, \text{step size})$ combination. The symbol \(\times\) marks settings that fail the stability guard, either because the run is non-finite or because the worst selection seed falls more than \(50\) nats below the best grid value; the color scale is clipped below best minus \(10\) nats.}
    \label{fig:heatmaps}
\end{figure}

\subsection{Empirical Details for \Cref{sec:Application: Census Binary Data}}\label{sec:Empirical Details for census Data}

\textbf{Census Experiment: $\mathbf{d=9}$}

Across the 10 BW-IW-ELBO runs, the $H$ eigenvalue clipping safeguard was activated on average in $12.9\%$ of iterations, with a range from $4.7\%$ to $41.9\%$. The median activation rate was $9.75\%$. The final minimum eigenvalue of $\Sigma$ was approximately $7.50 \times 10^{-5}$ in all runs, while the smallest minimum eigenvalue encountered during training across the ten runs ranged from $1.37 \times 10^{-6}$ to $5.36 \times 10^{-6}$.

The prior is $z \sim\mathcal{N}(0,10I_{9})$. All variational methods are initialized at
$m_0 = 0 \in \mathbb{R}^{9}$ and
$\Sigma_0 = 5I_{9}$. All methods are run for $1000$ iterations and averaged over $10$ independent training seeds. Step sizes were selected in a preliminary grid search using the highest exponentially smoothed moving-average objective over the same $1000$-iteration horizon, with smoothing window size $50$. The selected step sizes are $\eta=10^{-3}$ for BW-IW-ELBO, $\eta=10^{-4}$ for FB-GVI, $\mathrm{lr}=10^{-2}$ for Euclidean IW-ELBO, and $\mathrm{lr}=10^{-2}$ for Euclidean ELBO. For IW-ELBO methods, the moving-average criterion is applied to the IW-ELBO. For standard-ELBO methods, it is applied to the standard ELBO.

For the MCMC reference, we use a random-walk Metropolis-Hastings sampler with proposal
$$z' \sim \mathcal{N}\!\left(z_t,\ c\,(X^\top X)^{-1}\right), \qquad c = 3.0$$
The first chain is initialized at
$
\widehat{z}_{\mathrm{OLS}}
=
(X^{\top}X)^{-1}X^{\top}y,
$
and the remaining chains are initialized at perturbed versions of this point. We run four independent chains for $250{,}000$ iterations each, discarding the first $50{,}000$ iterations as burn-in and retaining $200{,}000$ samples per chain. The observed acceptance rates for the four chains are $48.34\%$, $48.81\%$, $48.42\%$, and $48.73\%$, respectively. Convergence is assessed using the Gelman--Rubin diagnostic, with $\widehat{R}_{\max} = 1.0004$. Trace plots and autocorrelations are reported in Figures \ref{fig:mcmc_trace} and Figure \ref{fig:mcmc_auto}.

\begin{figure}
    \centering
    \includegraphics[width=0.7\linewidth]{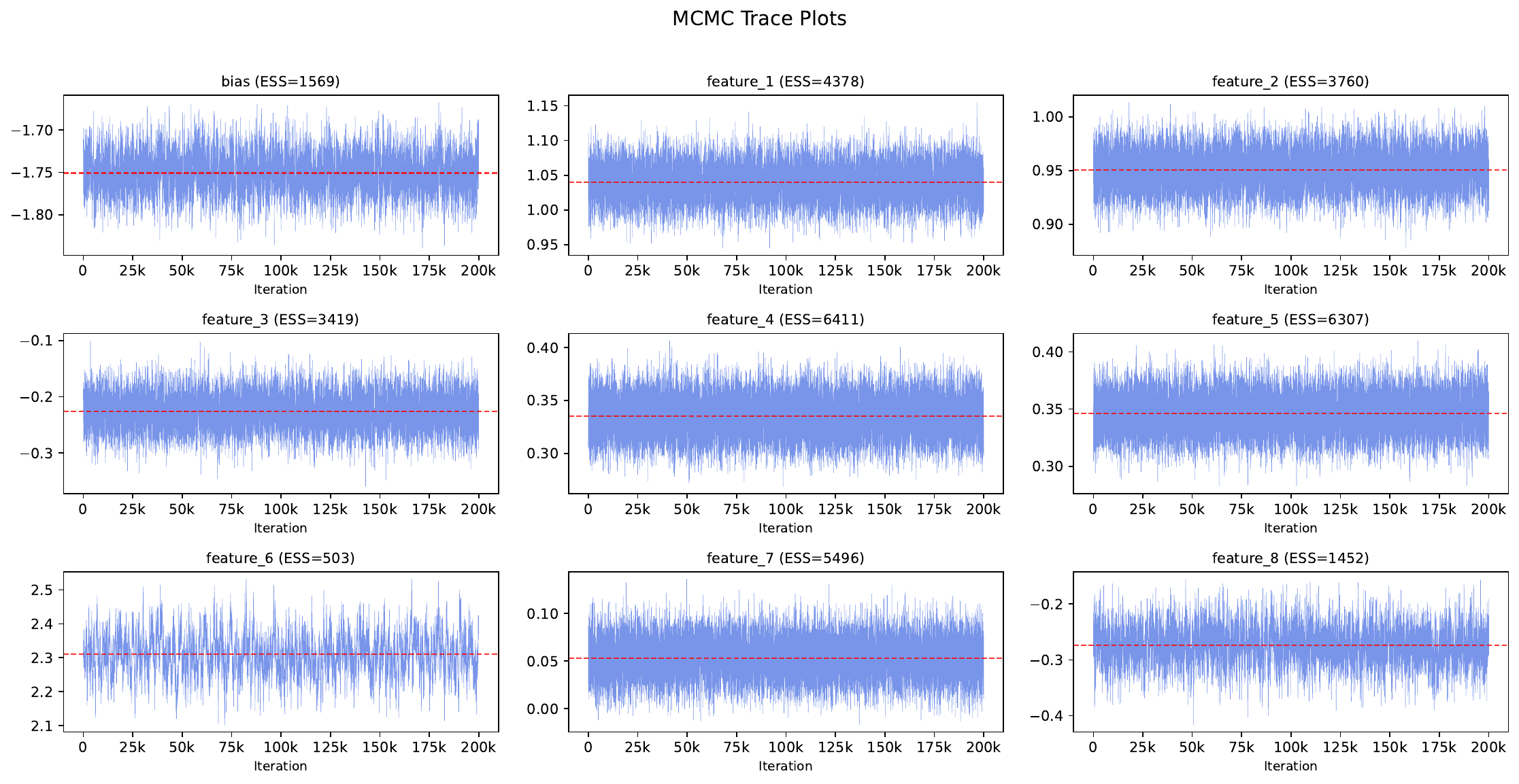}
    \caption{MCMC trace plots, which visually confirms that the MCMC chain has converged.}
    \label{fig:mcmc_trace}
\end{figure}

\begin{figure}
    \centering
    \includegraphics[width=0.7\linewidth]{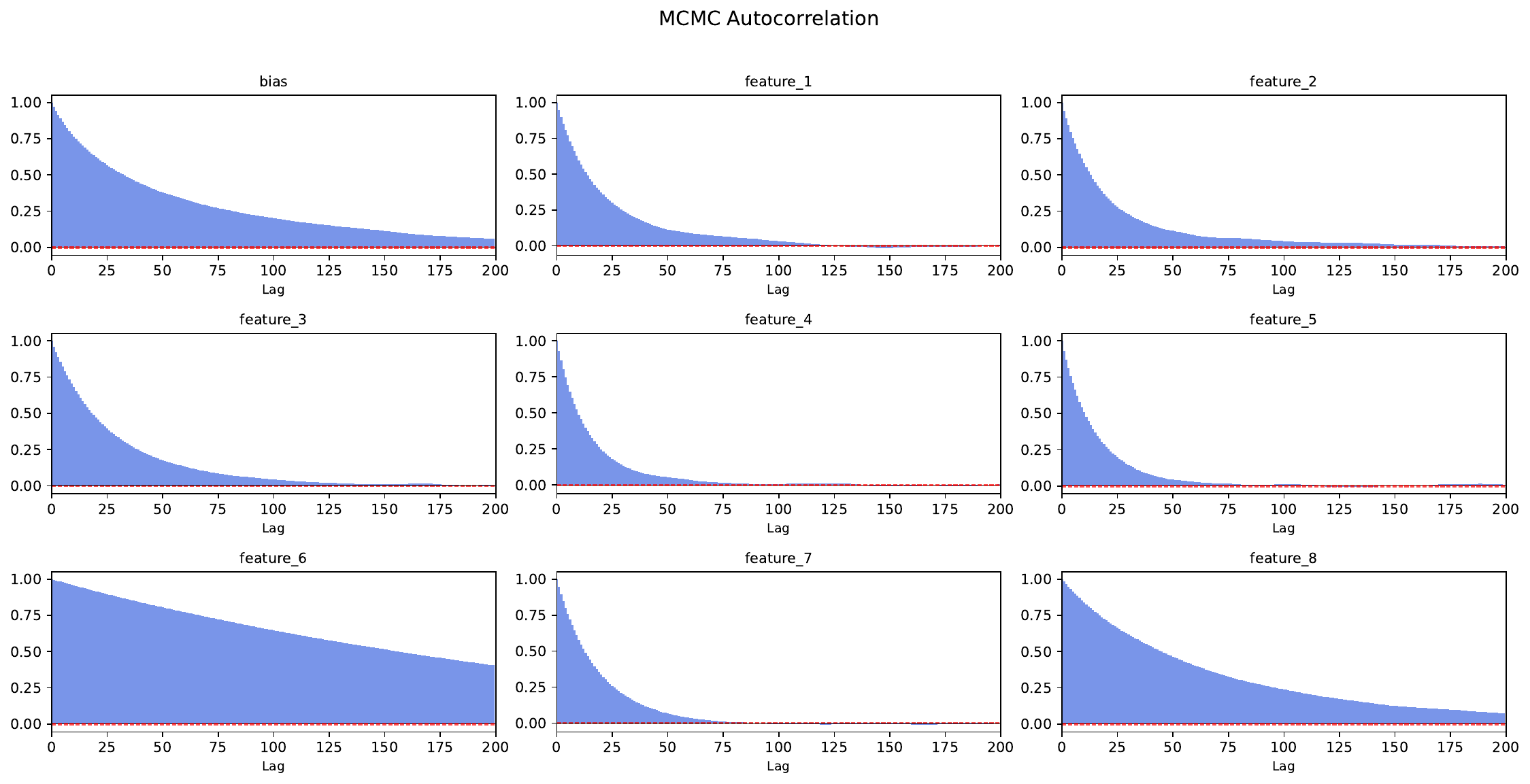}
    \caption{MCMC Autocorrelation}
    \label{fig:mcmc_auto}
\end{figure}

\begin{figure}
    \centering
    \includegraphics[width=1\linewidth]{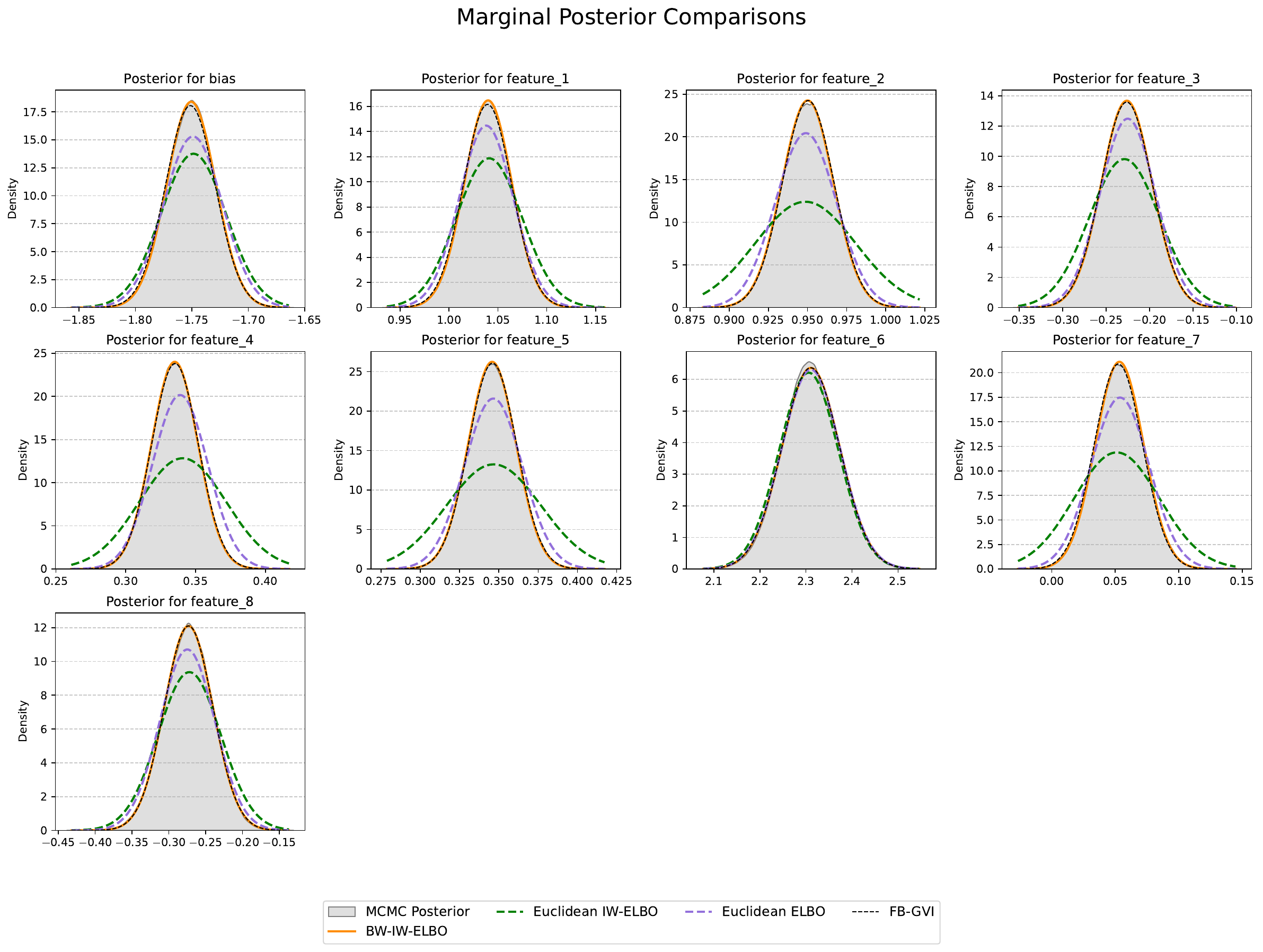}
    \caption{Marginal posteriors estimated by various VI methods and MCMC. BW-IW-ELBO estimates closely match those of MCMC.}
    \label{fig:marginal_post}
\end{figure}

\textbf{High-Dimensional Census Experiment: $\mathbf{d=96}$}

We repeat the Census experiment using the full feature set. The original design matrix contains $104$ columns, including the intercept. All non-intercept columns are standardized using $z$-scores, while the intercept is left unscaled. Since the one-hot encoded design contains exact linear dependencies, we apply column-pivoted QR after standardization and keep a maximal linearly independent subset of columns, forcing the intercept to be retained. This gives a full-rank design matrix with $d=96$ and condition number $8.0$.

The prior is $z \sim\mathcal{N}(0,10I_{96})$. All variational methods are initialized at
$m_0 = 0 \in \mathbb{R}^{96}$ and
$\Sigma_0 = 5I_{96}$. All methods are run for $1000$ iterations and averaged over $10$ independent training seeds, with smoothing window size $50$. The selected step sizes are $\eta=10^{-3}$ for BW-IW-ELBO, $\eta=5\times10^{-5}$ for FB-GVI, $\mathrm{lr}=10^{-2}$ for Euclidean IW-ELBO and Euclidean ELBO. Across the ten BW-IW-ELBO runs, eigenvalue clipping was activated in $18.7\%$ of iterations on average, ranging from $13.6\%$ to $26.9\%$. 

Because the $d=96$ posterior is sharply concentrated, we use a Laplace approximation as the reference posterior.  The Laplace approximation is appropriate here because the Bayesian logistic-regression log posterior is strictly concave, and the large ratio $N/d \approx 470$ places the posterior in a near-Gaussian large-sample regime. The MAP is computed by Newton's method on the analytic gradient and Hessian of
the log posterior. Writing $p_z = \sigma(Xz)$ and using a $\mathcal{N}(0,\sigma^2 I)$
prior with $\sigma^2 = 10$, the posterior precision at $z$ is
$$
\Lambda(z) = X^\top \operatorname{diag}\!\big(p_z \odot (1-p_z)\big) X
+ \sigma^{-2} I_{96},
$$
and the Laplace reference is
$p_{\mathrm{Lap}} = \mathcal{N}\!\big(\widehat z_{\mathrm{MAP}},\,
\Lambda(\widehat z_{\mathrm{MAP}})^{-1}\big)$.
Initialized at $z^{(0)} = 0$ and capped at $100$ iterations, Newton's method
converged after $9$ iterations---stopping when the maximum absolute Newton
step fell below $10^{-10}$---with a final maximum absolute gradient component of $1.22\times 10^{-12}$,
confirming an essentially exact MAP. The marginal posterior standard deviations
range from $0.0123$ to $0.0645$, and the Laplace covariance has condition
number $74.09$.

For each fitted Gaussian approximation $q=\mathcal{N}(m,\Sigma)$, we compute the final ELBO, normalized effective sample size and forward KL. The KL values in Table~\ref{tab:census-highdim} are Gaussian--Gaussian divergences relative to the Laplace reference. Specifically, for
$p=\mathcal{N}(m_p,\Sigma_p)$ and $q=\mathcal{N}(m_q,\Sigma_q)$,
$$
\mathrm{KL}(p\,\|\,q)
=
\frac{1}{2}
\left[
\operatorname{tr}(\Sigma_q^{-1}\Sigma_p)
+
(m_q-m_p)^\top \Sigma_q^{-1}(m_q-m_p)
- d
+
\log\frac{\det \Sigma_q}{\det \Sigma_p}
\right]
$$
Here $p$ denotes the Laplace reference and $q$ denotes the fitted variational Gaussian.

\end{document}